\documentclass{statsoc}

\usepackage[margin=0.9in]{geometry}
\usepackage{amsmath,amsfonts,amssymb}
\usepackage{setspace}
\usepackage{times}
\usepackage{graphicx}
\usepackage{comment}

\def\inprob{\stackrel{p}{\rightarrow}}
\def\indist{\stackrel{d}{\rightarrow}}

\def\T{{ \mathrm{\scriptscriptstyle T} }}

\def\expit{\text{expit}}
\newcommand{\var}{\text{var}}

\newcommand{\Pb}{\mathbb{P}}
\newcommand{\Pn}{\mathbb{P}_n}
\newcommand{\Gn}{\mathbb{G}_n}

\newcommand{\E}{\mathbb{E}}
\newcommand{\R}{\mathbb{R}}
\DeclareMathOperator*{\argmin}{arg\,min}
\DeclareSymbolFont{bbold}{U}{bbold}{m}{n}
\DeclareSymbolFontAlphabet{\mathbbold}{bbold}

\newcommand{\bZ}{\mathbf{Z}}
\newcommand{\bz}{\mathbf{z}}
\newcommand{\bL}{\mathbf{L}}
\newcommand{\bl}{\mathbf{l}}
\newcommand{\bg}{\mathbf{g}}

\usepackage{amsthm}
\newtheorem{theorem}{Theorem}
\newtheorem{lemma}{Lemma}
\newtheorem{assumption}{Assumption}

\title[Continuous treatment effects]{Nonparametric methods for doubly robust estimation of continuous treatment effects}
\author[Kennedy {\it et al.}]{Edward H. Kennedy, Zongming Ma, Matthew D. McHugh, and \\ Dylan S. Small}
\email{kennedye@mail.med.upenn.edu}
\address{University of Pennsylvania, Philadelphia, USA.}

\allowdisplaybreaks

\begin{document}

\singlespace

\begin{abstract}
Continuous treatments (e.g., doses) arise often in practice, but many available causal effect estimators are limited by either requiring parametric models for the effect curve, or by not allowing doubly robust covariate adjustment. We develop a novel kernel smoothing approach that requires only mild smoothness assumptions on the effect curve, and still allows for misspecification of either the treatment density or outcome regression. We derive asymptotic properties and give a procedure for data-driven bandwidth selection. The methods are illustrated via simulation and in a study of the effect of nurse staffing on hospital readmissions penalties. 
\end{abstract}

\section{Introduction}

Continuous treatments or exposures (such as dose, duration, and frequency) arise very often in practice, especially in observational studies. Importantly, such treatments lead to effects that are naturally described by curves (e.g., dose-response curves) rather than scalars, as might be the case for binary treatments. Two major methodological challenges in continuous treatment settings are (1) to allow for flexible estimation of the dose-response curve (for example to discover underlying structure without imposing a priori shape restrictions), and (2) to properly adjust for high-dimensional confounders (i.e., pre-treatment covariates related to treatment assignment and outcome).

Consider a recent example involving the Hospital Readmissions Reduction Program, instituted by the Centers for Medicare \& Medicaid Services in 2012, which aimed to reduce preventable hospital readmissions by penalizing hospitals with excess readmissions.  \citet{mchugh2013hospitals} were interested in whether nurse staffing (measured in nurse hours per patient day) affected hospitals' risk of excess readmissions penalty. The left panel of Figure \ref{fig:app} shows data for 2976 hospitals, with nurse staffing (the `treatment') on the x-axis, whether each hospital was penalized (the outcome) on the y-axis, and a loess curve fit to the data (without any adjustment). One way to characterize effects is to imagine setting all hospitals' nurse staffing to the same level, and seeing if changes in this level yield changes in excess readmissions risk. Such questions cannot be answered by simply comparing hospitals' risk of penalty across levels of nurse staffing, since hospitals differ in many important ways that could be related to both nurse staffing and excess readmissions (e.g., size, location, teaching status, among many other factors). The right panel of Figure \ref{fig:app} displays the extent of these hospital differences, showing for example that hospitals with more nurse staffing are also more likely to be high-technology hospitals and see patients with higher socioeconomic status. To correctly estimate the effect curve, and fairly compare the risk of readmissions penalty at different nurse staffing levels, one must adjust for hospital characteristics  appropriately. 

\begin{figure}[h!]
\includegraphics[width=\textwidth]{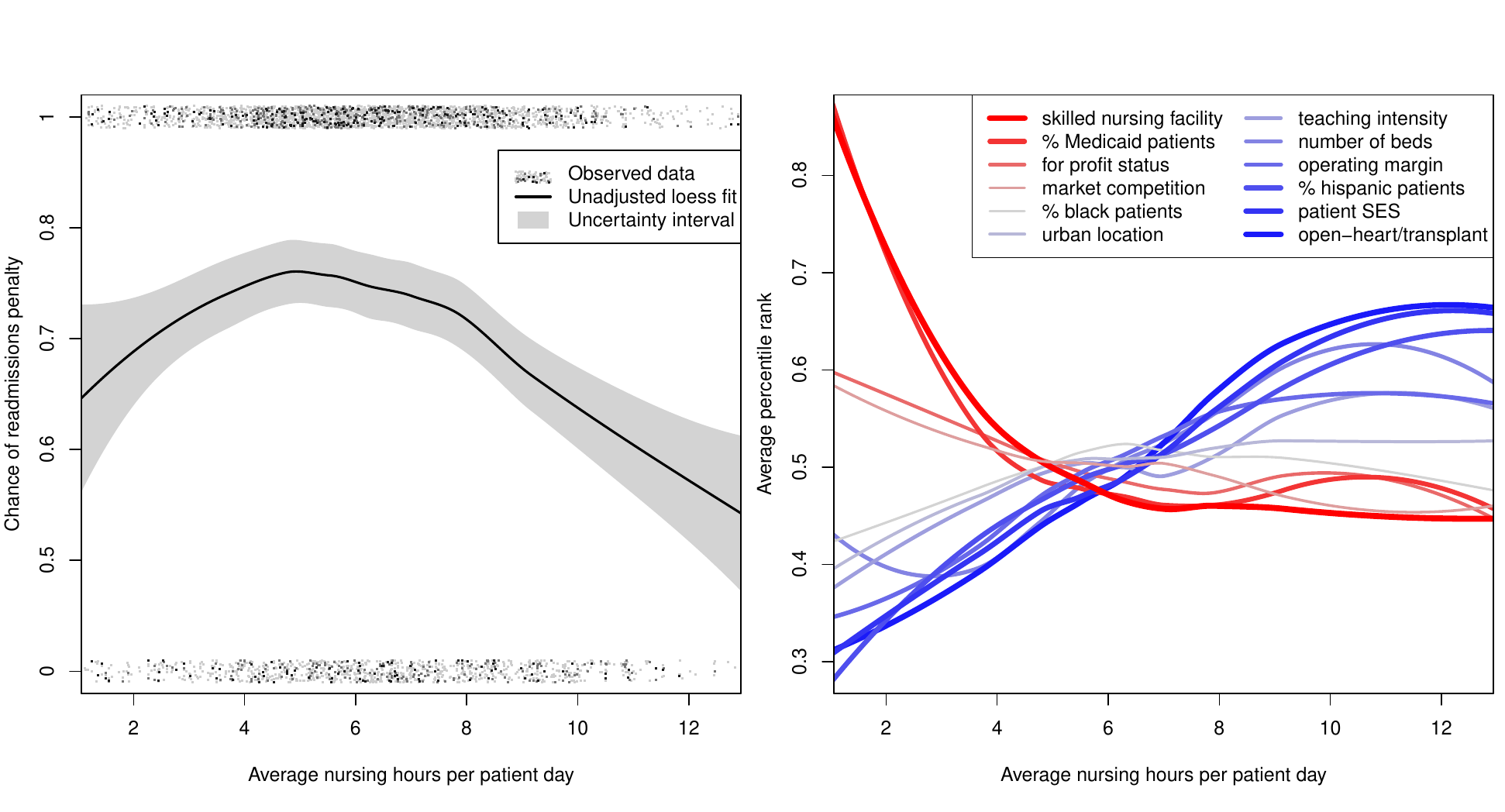}
\caption{Left panel: Observed treatment and outcome data with unadjusted loess fit. Right panel: Average covariate value as a function of exposure, after transforming to percentiles to display on common scale.}
\label{fig:app}
\end{figure}

In practice, the most common approach for estimating continuous treatment effects is based on regression modeling of how the outcome relates to covariates and treatment (e.g., \citet{imbens2004nonparametric}, \citet{hill2011bayesian}). However, this approach relies entirely on correct specification of the outcome model, does not incorporate available information about the treatment mechanism, and is sensitive to the curse of dimensionality by inheriting the rate of convergence of the outcome regression estimator.  \citet{hirano2004propensity}, \citet{imai2004causal}, and \citet{galvao2015uniformly} adapted propensity score-based approaches to the continuous treatment setting, but these similarly rely on correct specification of at least a model for treatment (e.g., the conditional treatment density).

In contrast, semiparametric doubly robust estimators \citep{robins2001inference,van2003unified} are based on modeling both the treatment and outcome processes and, remarkably, give consistent estimates of effects as long as one of these two nuisance processes is modeled well enough (not necessarily both). Beyond giving two independent chances at consistent estimation, doubly robust methods can also attain faster rates of convergence than their nuisance (i.e., outcome and treatment process) estimators when both models are consistently estimated; this makes them less sensitive to the curse of dimensionality and can allow for inference even after using flexible machine learning-based adjustment. However, standard semiparametric doubly robust methods for dose-response estimation rely on parametric models for the effect curve, either by explicitly assuming a parametric dose-response curve \citep{robins2000marginal,van2003unified}, or else by projecting the true curve onto a parametric working model \citep{neugebauer2007nonparametric}. Unfortunately, the first approach can lead to substantial bias under model misspecification, and the second can be of limited practical use if the working model is far away from the truth. 

Recent work has extended semiparametric doubly robust methods to more complicated nonparametric and high-dimensional settings. In a foundational paper, \citet{van2003cross} proposed a powerful cross-validation framework for estimator selection in general censored data and causal inference problems. Their empirical risk minimization approach allows for global nonparametric modeling in general semiparametric settings involving complex nuisance parameters. For example, \citet{diaz2013targeted} considered global modeling in the dose-response curve setting, and developed a doubly robust substitution estimator of risk. In nonparameric problems it is also important to consider non-global learning methods, e.g., via local and penalized modeling \citep{gyorfi2002distribution}. \citet{rubin2005general, rubin2006doubly, rubin2006extending} proposed extensions to such paradigms in numerous important problems, but the former considered weighted averages of dose-response curves and the latter did not consider doubly robust estimation.

In this paper we present a new approach for causal dose-response estimation that is doubly robust without requiring parametric assumptions, and which can naturally incorporate general machine learning methods. The approach is motivated by semiparametric theory for a particular stochastic intervention effect and a corresponding doubly robust mapping. Our method has a simple two-stage implementation that is fast and easy to use with standard software: in the first stage a pseudo-outcome is constructed based on the doubly robust mapping, and in the second stage the pseudo-outcome is regressed on treatment via off-the-shelf nonparametric regression and machine learning tools. We provide asymptotic results for a kernel version of our approach under weak assumptions, which only require mild smoothness conditions on the effect curve and allow for flexible data-adaptive estimation of relevant nuisance functions. We also discuss a simple method for bandwidth selection based on cross-validation. The methods are illustrated via simulation, and in the study discussed earlier about the effect of hospital nurse staffing on excess readmission penalties.

\section{Background}

\subsection{Data and notation}

Suppose we observe an independent and identically distributed sample $(\bZ_1,...,\bZ_n)$ where $\bZ=(\bL,A,Y)$ has support $\mathcal{Z}=(\mathcal{L} \times \mathcal{A} \times \mathcal{Y})$. Here $\bL$ denotes a vector of covariates, $A$ a continuous treatment or exposure, and $Y$ some outcome of interest. We characterize causal effects using potential outcome notation \citep{rubin1974estimating}, and so let $Y^a$ denote the potential outcome that would have been observed under treatment level $a$. 

We denote the distribution of $\bZ$ by $P$, with density $p(\bz)=p(y \mid \bl,a) p(a \mid \bl) p(\bl)$ with respect to some dominating measure. We let $\Pn$ denote the empirical measure so that empirical averages $n^{-1} \sum_i f(\bZ_i)$ can be written as $\Pn\{f(\bZ)\}=\int f(\bz) d\Pn(\bz)$. To simplify the presentation we denote the mean outcome given covariates and treatment with $\mu(\bl,a)=\E(Y \mid \bL=\bl,A=a)$, denote the conditional treatment density given covariates with $\pi(a \mid \bl)=\frac{\partial}{\partial a} P(A \leq a \mid \bL=\bl)$, and denote the marginal treatment density with $\varpi(a)=\frac{\partial}{\partial a}P(A\leq a)$. Finally, we use $||f||=\{ \int f(\bz)^2 dP(\bz) \}^{1/2}$ to denote the $L_2(P)$ norm, and we use $||f||_{\mathcal{X}}=\sup_{x \in \mathcal{X}} | f(x)|$ to denote the uniform norm of a generic function $f$ over $x \in \mathcal{X}$.

\subsection{Identification}

In this paper our goal is to estimate the effect curve $\theta(a) = \E(Y^a)$. Since this quantity is defined in terms of potential outcomes that are not directly observed, we must consider assumptions under which it can be expressed in terms of observed data. A full treatment of identification in the presence of continuous random variables was given by  \citet{gill2001causal}; we refer the reader there for details. The assumptions most commonly employed for identification are as follows (the following must hold for any $a \in \mathcal{A}$ at which $\theta(a)$ is to be identified).
\begin{assumption}
Consistency: $A=a$ implies $Y=Y^a$.
\end{assumption}
\begin{assumption}
Positivity: $\pi(a \mid \bl) \geq \pi_{min} > 0$ for all $\bl \in \mathcal{L}$.
\end{assumption}
\begin{assumption}
Ignorability: $\E(Y^a \mid \bL,A) = \E(Y^a \mid \bL)$.
\end{assumption}

Assumptions 1--3 can all be satisfied by design in randomized trials, but in observational studies they may be violated and are generally untestable. The consistency assumption ensures that potential outcomes are defined uniquely by a subject's own treatment level and not others' levels (i.e., no interference), and also not by the way treatment is administered (i.e., no different versions of treatment). Positivity says that treatment is not assigned deterministically, in the sense that every subject has some chance of receiving treatment level $a$, regardless of covariates; this can be a particularly strong assumption with continuous treatments. Ignorability says that the mean potential outcome under level $a$ is the same across treatment levels once we condition on covariates (i.e., treatment assignment is unrelated to potential outcomes within strata of covariates), and requires sufficiently many relevant covariates to be collected. Using the same logic as with discrete treatments, it is straightforward to show that under Assumptions 1--3 the effect curve $\theta(a)$ can be identified with observed data as
\begin{equation}
\theta(a) = \E\{ \mu(\bL,a) \} = \int_{\mathcal{L}} \mu(\bl,a) \ dP(\bl) .
\label{eq:ident1}
\end{equation}
Even if we are not willing to rely on Assumptions 1 and 3, it may often still be of interest to estimate $\theta(a)$ as an adjusted measure of association, defined purely in terms of observed data.

\section{Main results}

In this section we develop doubly robust estimators of the effect curve $\theta(a)$ without relying on parametric models. First we describe the logic behind our proposed approach, which is based on finding a doubly robust mapping whose conditional expectation given treatment equals the effect curve of interest, as long as one of two nuisance parameters is correctly specified. To find this mapping, we derive a novel efficient influence function for a stochastic intervention parameter. Our proposed method is based on regressing this doubly robust mapping on treatment using off-the-shelf nonparametric regression and machine learning methods. We derive asymptotic properties for a particular version of this approach based on local-linear kernel smoothing. Specifically, we give conditions for consistency and asymptotic normality, and describe how to use cross-validation to select the bandwidth parameter in practice.

\subsection{Setup and doubly robust mapping}

If $\theta(a)$ is assumed known up to a finite-dimensional parameter, for example $\theta(a)=\psi_0+\psi_1 a$ for $(\psi_0,\psi_1) \in \R^2$, then standard semiparametric theory can be used to derive the efficient influence function, from which one can obtain the efficiency bound and an efficient estimator \citep{bickel1993efficient,van2003unified,tsiatis2006semiparametric}. However, such theory is not directly available if we only  assume, for example, mild smoothness conditions on $\theta(a)$ (e.g., differentiability). This is due to the fact that without parametric assumptions $\theta(a)$ is not pathwise differentiable, and root-n consistent estimators do not exist \citep{bickel1993efficient,diaz2013targeted}. In this case there is no developed efficiency theory.

To derive doubly robust estimators for $\theta(a)$ without relying on parametric models, we adapt semiparametric theory in a novel way similar to the approach of \citet{rubin2005general, rubin2006doubly}. Our goal is to find a function $\xi(\bZ;\pi,\mu)$ of the observed data $\bZ$ and nuisance functions $(\pi,\mu)$ such that
$$ \E\{ \xi(\bZ; \overline\pi,\overline\mu) \mid A=a\} = \theta(a) $$
if either $\overline\pi=\pi$ or $\overline\mu=\mu$ (not necessarily both). Given such a mapping, off-the-shelf nonparametric regression and machine learning methods could be used to estimate $\theta(a)$ by regressing $\xi(\bZ;\hat\pi,\hat\mu)$ on treatment $A$, based on estimates $\hat\pi$ and $\hat\mu$. 

This doubly robust mapping is intimately related to semiparametric theory and especially the efficient influence function for a particular parameter. Specifically, if $\E\{ \xi(\bZ; \overline\pi,\overline\mu) \mid A=a\} = \theta(a)$ then it follows that  $\E\{ \xi(\bZ; \overline\pi,\overline\mu) \}=\psi$ for
\begin{equation}
 \psi=\int_\mathcal{A} \int_\mathcal{L} \mu(\bl,a) \varpi(a) \ dP(\bl) \ da . \label{eq:psi}
\end{equation}
This indicates that a natural candidate for the unknown mapping $\xi(\bZ;\pi,\mu)$ would be a component of the efficient influence function for the parameter $\psi$, since for regular parameters such as $\psi$ in semi- or non-parametric models, the efficient influence function $\phi(\bZ; \pi,\mu)$ will be doubly robust in the sense that $\E\{\phi(\bZ;\overline\pi,\overline\mu)\}=0$, if either $\overline\pi=\pi$ or $\overline\mu=\mu$ \citep{robins2001inference,van2003unified}. This implies $\E\{\phi(\bZ;\pi,\mu)\} = \E\{\xi(\bZ;\pi,\mu)-\psi\} = 0$ so that $\E\{\xi(\bZ;\overline\pi,\overline\mu)\}=\psi$ if either $\overline\pi=\pi$ or $\overline\mu=\mu$. This kind of logic was first used by \citet{rubin2005general, rubin2006doubly} for full data parameters that are functions of covariates rather than treatment (i.e., censoring) variables.

The parameter $\psi$ is also of interest in its own right. In particular, it represents the average outcome under an intervention that randomly assigns treatment based on the density $\varpi$ (i.e., a randomized trial). Thus comparing the value of this parameter to the average observed outcome provides a test of treatment effect; if the values differ significantly, then there is evidence that the observational treatment mechanism impacts outcomes for at least some units. Stochastic interventions were discussed by \citet{diaz2012population}, for example, but the efficient influence function for $\psi$ has not been given before under a nonparametric model. Thus in Theorem 1 below we give the efficient influence function for this parameter respecting the fact that the marginal density $\varpi$ is unknown.

\begin{theorem}
Under a nonparametric model, the efficient influence function for $\psi$ defined in \eqref{eq:psi} is  $\xi(\bZ;\pi,\mu) - \psi + \int_\mathcal{A} \{\mu(\bL,a) - \int_\mathcal{L} \mu(\bl,a) dP(\bl) \} \varpi(a) da$, where
$$ \xi(\bZ;\pi,\mu) = \frac{Y-\mu(\bL,A)}{\pi(A \mid \bL)} \int_\mathcal{L} \pi(A \mid \bl) \ dP(\bl) + \int_\mathcal{L} \mu(\bl,A) \ dP(\bl) . $$
\end{theorem}

A proof of Theorem 1 is given in the Appendix (Section 2). Importantly, we also prove that the function $\xi(\bZ;\pi,\mu)$ satisfies its desired double robustness property, i.e., that $\E\{\xi(\bZ;\overline\pi,\overline\mu) \mid A=a\}=\theta(a)$ if either $\overline\pi=\pi$ or $\overline\mu=\mu$. As mentioned earlier, this motivates estimating the effect curve $\theta(a)$ by estimating the nuisance functions $(\pi,\mu)$, and then regressing the estimated pseudo-outcome 
$$ \hat\xi(\bZ;\hat\pi,\hat\mu) = \frac{Y-\hat\mu(\bL,A)}{\hat\pi(A \mid \bL)} \int_\mathcal{L} \hat\pi(A \mid \bl) \ d\Pn(\bl) + \int_\mathcal{L} \hat\mu(\bl,A) \ d\Pn(\bl) $$
on treatment $A$ using off-the-shelf nonparametric regression or machine learning methods.  In the next subsection we describe our proposed approach in more detail, and analyze the properties of an estimator based on kernel estimation.

\subsection{Proposed approach}

In the previous subsection we derived a doubly robust mapping $\xi(\bZ;\pi,\mu)$ for which $\E\{\xi(\bZ;\overline\pi,\overline\mu) \mid A=a\}=\theta(a)$ as long as either $\overline\pi=\pi$ or $\overline\mu=\mu$. This indicates that doubly robust nonparametric estimation of $\theta(a)$ can proceed with a simple two-step procedure, where both steps can be accomplished with flexible machine learning. To summarize, our proposed method is:
\begin{enumerate}
\item[1.] Estimate nuisance functions $(\pi,\mu)$ and obtain predicted values.
\item[2.] Construct pseudo-outcome $\hat\xi(\bZ;\hat\pi,\hat\mu)$ and regress on treatment variable $A$.
\end{enumerate}
We give sample code implementing the above in the Appendix (Section 9).

In what follows we present results for an estimator that uses kernel smoothing in Step 2. Such an approach is related to kernel approximation of a full-data parameter in censored data settings. \citet{robins2001inference} gave general discussion and considered density estimation with missing data,  while \citet{van1998locally}, \citet{van2001inference},  and \citet{van2006estimating} used the approach for current status survival analysis; \citet{wang2010nonparametric} used it implicitly for nonparametric regression with missing outcomes. 

As indicated above, however, a wide variety of flexible methods could be used in our Step 2, including local partitioning or nearest neighbor estimation, global series or spline methods with complexity penalties, or cross-validation-based combinations of methods, e.g., Super Learner \citep{van2007super}. In general we expect the results we report in this paper to hold for many such methods. To see why, let $\hat\theta$ denote the proposed estimator described above (based on some initial nuisance estimators $(\hat\pi,\hat\mu)$ and a particular regression method in Step 2), and let $\overline\theta$ denote an estimator based on an oracle version of the pseudo-outcome $\xi(\bZ;\overline\pi,\overline\mu)$ where $(\overline\pi,\overline\mu)$ are the unknown limits to which the estimators $(\hat\pi,\hat\mu)$ converge. Then
$ || \hat\theta - \theta || \leq || \hat\theta - \overline\theta || + || \overline\theta - \theta || $, where the second term on the right can be analyzed with standard theory since $\overline\theta$ is a regression of a simple fixed function $\xi(\bZ;\overline\pi,\overline\mu)$ on $A$, and the first term will be small depending on the convergence rates of $\hat\pi$ and $\hat\mu$. A similar point was discussed by \citet{rubin2005general, rubin2006doubly}. 

The local linear kernel version of our estimator is $\hat\theta_h(a)=\bg_{ha}(a)^\T \boldsymbol{\hat\beta}_h(a)$, where $\bg_{ha}(t)=(1,\frac{t-a}{h})^\T$ and 
\begin{equation}
\boldsymbol{\hat\beta}_h(a) = \argmin_{\boldsymbol\beta \in \R^2} \ \Pn\left[ K_{ha}(A) \Big\{ \hat\xi(\bZ;\hat\pi,\hat\mu)  - \bg_{ha}(A)^\T \boldsymbol\beta \Big\}^2 \right] 
\end{equation}
for $K_{ha}(t)=h^{-1}K\{(t-a)/h\}$ with $K$ a standard kernel function (e.g., a symmetric probability density) and $h$ a scalar bandwidth parameter. This is a standard local linear kernel regression of $\hat\xi(\bZ;\hat\pi,\hat\mu)$ on $A$. For overviews of kernel smoothing see, e.g., \citet{fan1996local}, \citet{wasserman2006all}, and \citet{li2007nonparametric}. Under near-violations of positivity, the above estimator could potentially lie outside the range of possible values for $\theta(a)$ (e.g., if $Y$ is binary); thus we present a targeted minimum loss-based estimator (TMLE) in the Appendix (Section 4), which does not have this problem. Alternatively one could project onto a logistic model in (3).

\subsection{Consistency of kernel estimator}

In Theorem 2 below we give conditions under which the proposed kernel estimator $\hat\theta_h(a)$ is consistent for $\theta(a)$, and also give the corresponding rate of convergence. In general this result follows if the bandwidth decreases with sample size slowly enough, and if either of the nuisance functions $\pi$ or $\mu$ is estimated well enough (not necessarily both). The rate of convergence is a sum of two rates: one from standard nonparametric regression problems (depending on the bandwidth $h$), and another coming from estimation of the nuisance functions $\pi$ and $\mu$.

\begin{theorem}
Let $\overline\pi$ and $\overline\mu$ denote fixed functions to which $\hat\pi$ and $\hat\mu$ converge in the sense that $||\hat\pi-\overline\pi||_\mathcal{Z}=o_p(1)$ and $||\hat\mu-\overline\mu||_\mathcal{Z}=o_p(1)$, and let $a \in \mathcal{A}$ denote a point in the interior of the compact support $\mathcal{A}$ of $A$. Along with Assumption 2 (Positivity), assume the following:
\begin{enumerate}
\item Either $\overline\pi=\pi$ or $\overline\mu=\mu$.
\item The bandwidth $h=h_n$ satisfies $h \rightarrow 0$ and $nh^3 \rightarrow \infty$ as $n \rightarrow \infty$.
\item $K$ is a continuous symmetric probability density with support $[-1,1]$.
\item $\theta(a)$ is twice continuously differentiable, and both $\varpi(a)$ and the conditional density of $\xi(\bZ;\overline\pi,\overline\mu)$ given $A=a$ are continuous as functions of $a$.
\item The estimators $(\hat\pi,\hat\mu)$ and their limits $(\overline\pi,\overline\mu)$ are contained in uniformly bounded function classes with finite uniform entropy integrals (as defined in Section 5 of the Appendix), with $1/\hat\pi$ and $1/\overline\pi$ also uniformly bounded.
\end{enumerate}
Then 
$$ | \hat\theta_h(a) - \theta(a) | = O_p\left( \frac{1}{\sqrt{nh}} + h^2 + r_n(a) s_n(a) \right) $$
where 
$$ \sup_{t: |t-a|\leq h} || \hat\pi(t  \mid \bL)-\pi(t \mid \bL) || = O_p\Big(r(n)\Big) $$
$$ \sup_{t: |t-a|\leq h} || \hat\mu(\bL , t)-\mu(\bL, t) || = O_p\Big(s(n)\Big) $$
are the `local' rates of convergence of $\hat\pi$ and $\hat\mu$ near $A=a$.
\end{theorem}

A proof of Theorem 2 is given in the Appendix (Section 6). The required conditions are all quite weak. Condition (a) is arguably the most important of the conditions, and says that at least one of the estimators $\hat\pi$ or $\hat\mu$ must be consistent for the true $\pi$ or $\mu$ in terms of the uniform norm. Since only one of the nuisance estimators is required to be consistent (not both), Theorem 2 shows the double robustness of the proposed estimator $\hat\theta_h(a)$. Conditions (b), (c), and (d) are all common in standard nonparametric regression problems, while condition (e) involves the complexity of the estimators $\hat\pi$ and $\hat\mu$ (and their limits), and is a usual minimal regularity condition for problems involving nuisance functions. 

Condition (b) says that the bandwidth parameter $h$ decreases with sample size but not too quickly (so that $nh^3 \rightarrow \infty$). This is a standard requirement in local linear kernel smoothing \citep{fan1996local,wasserman2006all,li2007nonparametric}. Note that since $nh=nh^3 /h^2$, it is implied that $nh \rightarrow \infty$; thus one can view $nh$ as a kind of effective or local sample size. Roughly speaking, the bandwidth $h$ needs to go to zero in order to control bias, while the local sample size $nh$ (and $nh^3$) needs to go to infinity in order to control variance. We postpone more detailed discussion of the bandwidth parameter until a later subsection, where we detail how it can be chosen in practice using cross-validation. Condition (c) puts some minimal restrictions on the kernel function. It is clearly satisfied for most common kernels, including the uniform kernel $K(u)=I(|u|\leq 1)/2$, the Epanechnikov kernel $K(u)=(3/4)(1-u^2) I(|u|\leq 1)$, and a truncated version of the Gaussian kernel $K(u)= I(|u| \leq 1) \phi(u) /\{2\Phi(1)-1\}$ with $\phi$ and $\Phi$ the density and distribution functions for a standard normal random variable. Condition (d) restricts the smoothness of the effect curve $\theta(a)$, the density of $\varpi(a)$, and the conditional density given $A=a$ of the limiting pseudo-outcome $\xi(\bZ;\overline\pi,\overline\mu)$. These are standard smoothness conditions imposed in nonparametric regression problems. By assuming more smoothness of $\theta(a)$, bias-reducing (higher-order) kernels could achieve faster rates of convergence and even approach the parametric root-n rate (see for example \citet{fan1996local}, \citet{wasserman2006all}, and others).

Condition (e) puts a mild restriction on how flexible the nuisance estimators (and their corresponding limits) can be, although such uniform entropy conditions still allow for a wide array of data-adaptive estimators and, importantly, do not require the use of parametric models. \citet{andrews1994empirical} (Section 4), \citet{van1996weak} (Sections 2.6--2.7), and \citet{van2000asymptotic} (Examples 19.6--19.12) discuss a wide variety of function classes with finite uniform entropy integrals. Examples include standard parametric classes of functions indexed by Euclidean parameters (e.g., parametric functions satisfying a Lipschitz condition), smooth functions with uniformly bounded partial derivatives, Sobolev classes of functions, as well as convex combinations or Lipschitz transformations of any such sets of functions. The uniform entropy restriction in condition (e) is therefore not a very strong restriction in practice; however, it could be further weakened via sample splitting techniques (see Chapter 27 of \citet{van2011targeted}).

The convergence rate given in the result of Theorem 2 is a sum of two components. The first, $1/\sqrt{nh} + h^2$, is the rate achieved in standard nonparametric regression problems without nuisance functions. Note that if $h$ tends to zero slowly, then $1/\sqrt{nh}$ will tend to zero quickly but $h^2$ will tend to zero more slowly; similarly if $h$ tends to zero quickly, then $h^2$ will as well, but $1/\sqrt{nh}$ will tend to zero more slowly. Balancing these two terms requires $h \sim n^{-1/5}$ so that $1/\sqrt{nh} \sim h^2 \sim n^{-2/5}$. This is the optimal pointwise rate of convergence for standard nonparametric regression on a single covariate, for a twice continuously differentiable regression function. 

The second component, $r_n(a) s_n(a)$, is the product of the local rates of convergence (around $A=a$) of the nuisance estimators $\hat\pi$ and $\hat\mu$ towards their targets $\pi$ and $\mu$. Thus if the nuisance function estimates converge slowly (due to the curse of dimensionality), then the convergence rate of $\hat\theta_h(a)$ will also be slow. However, since the term is a product, we have two chances at obtaining fast convergence rates, showing the bias-reducing benefit of doubly robust estimators. The usual explanation of double robustness is that, even if $\hat\mu$ is misspecified so that $s_n(a)=O(1)$, then as long as $\hat\pi$ is consistent, i.e., $r_n(a)=o(1)$, we will still have consistency since $r_n(a) s_n(a)=o(1)$. But this idea also extends to settings when both $\hat\pi$ and $\hat\mu$ are consistent. For example suppose $h \sim n^{-1/5}$ so that $1/\sqrt{nh}+h^2 \sim n^{-2/5}$, and suppose $\hat\pi$ and $\hat\mu$ are locally consistent with rates $r_n(a)=n^{-2/5}$ and $s_n(a)=n^{-1/10}$. Then the product is $r_n(a) s_n(a) = O(n^{-1/2}) = o(n^{-2/5})$ and the contribution from the nuisance functions is asymptotically negligible, in the sense that the proposed estimator has the same convergence rate as an infeasible estimator with known nuisance functions. Contrast this with non-doubly-robust plug-in estimators whose convergence rate generally matches that of the nuisance function estimator, rather than being faster \citep{van2014higher}.

In Section 8 of the Appendix we give some discussion of uniform consistency, which, along with weak convergence, will be pursued in more detail in future work. 

\subsection{Asymptotic normality of kernel estimator}

In the next theorem we show that if one or both of the nuisance functions are estimated at fast enough rates, then the proposed estimator is asymptotically normal after appropriate scaling.

\pagebreak

\begin{theorem}
Consider the same setting as Theorem 2. Along with Assumption 2 (Positivity) and conditions (a)--(e) from Theorem 2, also assume that:
\begin{enumerate}
\item[(f)] The local convergence rates satisfy $r_n(a) s_n(a) = o_p(1/\sqrt{nh})$.
\end{enumerate}
Then
$$\sqrt{nh} \Big\{ \hat\theta_h(a) - \theta(a) + b_h(a) \Big\} \indist N\!\left(0, \ \frac{\sigma^2(a) \int K(u)^2 \ du}{\varpi(a)} \right) $$
where $b_h(a) = \theta''(a) (h^2 / 2) \int u^2 K(u) \ du + o(h^2)$, and
\begin{equation*}
\sigma^2(a) = \E\!\left[ \frac{\tau^2(\bL,a) + \{\mu(\bL,a)-\overline\mu(\bL,a)\}^2}{\{\overline\pi(a \mid \bL)/\overline\varpi(a)\}^2 / \{\pi(a \mid \bL)/\varpi(a)\}} \right] - \Big\{ \theta(a)-\overline{m}(a) \Big\}^2 
\end{equation*}
for $\tau^2(\bl,a) = \var(Y \mid \bL=\bl, A=a)$, $\overline\varpi(a)=\E\{\overline\pi(a \mid \bL)\}$, $\overline{m}(a)=\E\{\overline\mu(\bL,a)\}$.
\end{theorem}

The proof of Theorem 3 is given in the Appendix (Section 7). Conditions (a)--(e) are the same as in Theorem 2 and were discussed earlier. Condition (f) puts a restriction on the local convergence rates of the nuisance estimators. This will in general require at least some semiparametric modeling of the nuisance functions. Truly nonparametric estimators of $\pi$ and $\mu$ will typically converge at slow rates due to the curse of dimensionality, and will generally not satisfy the rate requirement in the presence of multiple continuous covariates. Condition (f) basically ensures that estimation of the nuisance functions is irrelevant asymptotically; depending on the specific nuisance estimators used, it could be possible to give weaker but more complicated conditions that allow for a non-negligible asymptotic contribution while still yielding asymptotic normality. 

Importantly, the rate of convergence required by condition (g) of Theorem 3 is slower than the root-n rate typically required in standard semiparametric settings where the parameter of interest is finite-dimensional and Euclidean. For example, in a standard setting where the support $\mathcal{A}$ is finite, a sufficient condition for yielding the requisite asymptotic negligibility for attaining efficiency is $r_n(a)=s_n(a)=o(n^{-1/4})$; however in our setting the weaker condition $r_n(a)=s_n(a)=o(n^{-1/5})$ would be sufficient if $h \sim n^{-1/5}$. Similarly, if one nuisance estimator $\hat\pi$ or $\hat\mu$ is computed with a correctly specified generalized additive model, then the other nuisance estimator would ony need to be consistent (without a rate condition). This is because, under regularity conditions and with optimal smoothing, a generalized additive model estimator converges at rate $O_p(n^{-2/5})$ \citep{horowitz2009semiparametric}, so that if the other nuisance estimator is merely consistent we have $r_n(a) s_n(a) = O(n^{-2/5}) o(1) = o(n^{-2/5})$, which satisfies condition (f) as long as $h \sim n^{-1/5}$. In standard settings such flexible nuisance estimation would make a non-negligible contribution to the limiting behavior of the estimator, preventing asymptotic normality and root-n consistency. 

Under the assumptions of Theorem 3, the proposed estimator is asymptotically normal after appropriate scaling and centering. However, the scaling is by the square root of the local sample size $\sqrt{nh}$ rather than the usual parametric rate $\sqrt{n}$. This slower convergence rate is a cost of making fewer assumptions (equivalently, the cost of better efficiency would be less robustness); thus we have a typical bias-variance trade-off. As in standard nonparametric regression, the estimator is consistent but not quite centered at $\theta(a)$; there is a bias term of order $O(h^2)$, denoted $b_h(a)$. In fact the estimator is centered at a smoothed version of the effect curve $\theta^*_h(a)=\bg_{ha}(a)^\T \boldsymbol\beta_h(a)=\theta(a)+b_h(a)$. This phenomenon is ubiquitous in nonparametric regression, and complicates the process of computing confidence bands. It is sometimes assumed that the bias term is $o(1/\sqrt{nh})$ and thus asymptotically negligible (e.g., by assuming $h=o(n^{-1/5})$ so that $nh^5 \rightarrow 0$); this is called undersmoothing and technically allows for the construction of valid confidence bands around $\theta(a)$. However,  there is little guidance about how to actually undersmooth in practice, so it is mostly a technical device. We follow \citet{wasserman2006all} and others by expressing uncertainty about the estimator $\hat\theta_h(a)$ using confidence intervals centered at the smoothed data-dependent parameter $\theta^*_h(a)$. For example, under the conditions of Theorem 3, pointwise Wald 95\% confidence intervals can be constructed with $\hat\theta_h(a) \pm 1.96 \hat\sigma/\sqrt{n}$, where $\hat\sigma^2$ is the $(1,1)$ element of the sandwich variance estimate $\Pn\{ \boldsymbol{\hat\varphi}_{ha}(\bZ)^{\otimes 2}\}$ based on the estimated efficient influence function for $\boldsymbol\beta_h(a)$ given by
\begin{align*}
\boldsymbol{\hat\varphi}_{ha}(\bZ) = \mathbf{\hat{D}}_{ha}^{-1} \bigg[ \bg_{ha}(A) &K_{ha}(A) \Big\{ \hat\xi(\bZ;\hat\pi,\hat\mu) - \bg_{ha}(A)^\T \boldsymbol{\hat\beta}_h(a) \Big\} \\
&+ \int_{\mathcal{A}} \bg_{ha}(t) K_{ha}(t) \Big\{ \hat\mu(\bL,t) - \hat{m}(t) \Big\} \hat\varpi(t) \ dt \bigg]
\end{align*}
for $\mathbf{\hat{D}}_{ha}=\Pn\{\bg_{ha}(A) K_{ha}(A) \bg_{ha}^\T \}$, $\hat{m}(t)=\Pn\{\hat\mu(\bL,t)\}$, $\hat\varpi(t)=\Pn\{\hat\pi(t \mid \bL)\}$.

\subsection{Data-driven bandwidth selection}

The choice of smoothing parameter is critical for any nonparametric method; too much smoothing yields large biases and too little yields excessive variance. In this subsection we discuss how to use cross-validation to choose the relevant bandwidth parameter $h$. In general the method we propose parallels those used in standard nonparametric regression settings, and can give similar optimality properties.

We can exploit the fact that our method can be cast as a non-standard nonparametric regression problem, and borrow from the wealth of literature on bandwidth selection there. Specifically, the logic behind Theorem 3 (i.e., that nuisance function estimation can be asymptotically irrelevant) can be adapted to the bandwidth selection setting, by treating the pseudo-outcome $\xi(\bZ;\hat\pi,\hat\mu)$ as known and using for example the bandwidth selection framework from \citet{hardle1988far}. These authors proposed a unified selection approach that includes generalized cross-validation, Akaike's information criterion, and leave-one-out cross-validation as special cases, and showed the asymptotic equivalence and optimality of such approaches. In our setting, leave-one-out cross-validation is attractive because of its computational ease. The simplest analog of leave-one-out cross-validation for our problem would be to select the optimal bandwidth from some set $\mathcal{H}$ with
$$ \hat{h}_{opt} = \argmin_{h \in \mathcal{H}} \ \sum_{i=1}^n  \left\{ \frac{ \hat\xi(\bZ_i;\hat\pi,\hat\mu)  - \hat\theta_h(A_i) }{1-\hat{W}_h(A_i)} \right\}^2, $$
where $\hat{W}_h(a_i)= (1,0) \Pn\{ \bg_{ha_i}(A) K_{ha_i}(A) \bg_{ha_i}(A)^\T \}^{-1} (1,0)^\T h^{-1} K(0)$  
is the $i^{th}$ diagonal of the so-called smoothing or hat matrix. The properties of this approach can be derived using logic similar to that in Theorem 3, e.g., by adapting results from \citet{li2004cross}. Alternatively one could split the sample, estimate the nuisance functions in one half, and then do leave-one-out cross-validation in the other half, treating the pseudo-outcomes estimated in the other half as known. We expect these approaches to be asymptotically equivalent to an oracle selector.

An alternative option would be to use the $k$-fold cross-validation approach of \citet{van2003cross} or \citet{diaz2013targeted}. This would entail randomly splitting the data into $k$ parts, estimating the nuisance functions and the effect curve on $(k-1)$ training folds, using these estimates to compute measures of risk on the $k^{th}$ test fold, and then repeating across all $k$ folds and averaging the risk estimates. One would then repeat this process for each bandwidth value $h$ in some set $\mathcal{H}$, and pick that which minimized the estimated cross-validated risk. \citet{van2003cross} gave finite-sample and asymptotic results showing that the resulting estimator behaves similarly to an oracle estimator that minimizes the true unknown cross-validated risk. Unfortunately this cross-validation process can be more computationally intensive than the above leave-one-out method, especially if the nuisance functions are estimated with flexible computation-heavy methods. However this approach will be crucial when incorporating general machine learning and moving beyond linear kernel smoothers.


\section{Simulation study}

We used simulation to examine the finite-sample properties of our proposed methods. Specifically we simulated from a model with normally distributed covariates
$$ \bL=(L_1,...,L_4)^\T \sim N(0,\mathbf{I}_4) , $$
Beta distributed exposure
\begin{equation*}
\begin{gathered}
(A/20) \mid \bL \sim \text{Beta}\{ \lambda(\bL), 1-\lambda(\bL)\} , \\
\text{logit} \ \lambda(\bL) = -0.8+0.1 {L}_1 + 0.1 {L}_2 - 0.1 {L}_3 + 0.2 {L}_4 , \\
\end{gathered}
\end{equation*}
and a binary outcome 
\begin{equation*}
\begin{gathered}
Y \mid \bL, A \sim \text{Bernoulli}\{\mu(\bL,A)\} , \\
\text{logit} \ \mu(\bL,A)=  1 + (0.2,0.2,0.3,-0.1) \bL + A (0.1 - 0.1 {L}_1 + 0.1 {L}_3 - 0.13^2 A^2) .
\end{gathered}
\end{equation*}
The above setup roughly matches the data example from the next section. Figure \ref{fig:simplot} shows a plot of the effect curve $\theta(a)=\E\{\mu(\bL,a)\}$ induced by the simulation setup, along with treatment versus outcome data for one simulated dataset (with $n=1000$).

\begin{figure}[h!]
\centering
\includegraphics[width=0.8\textwidth]{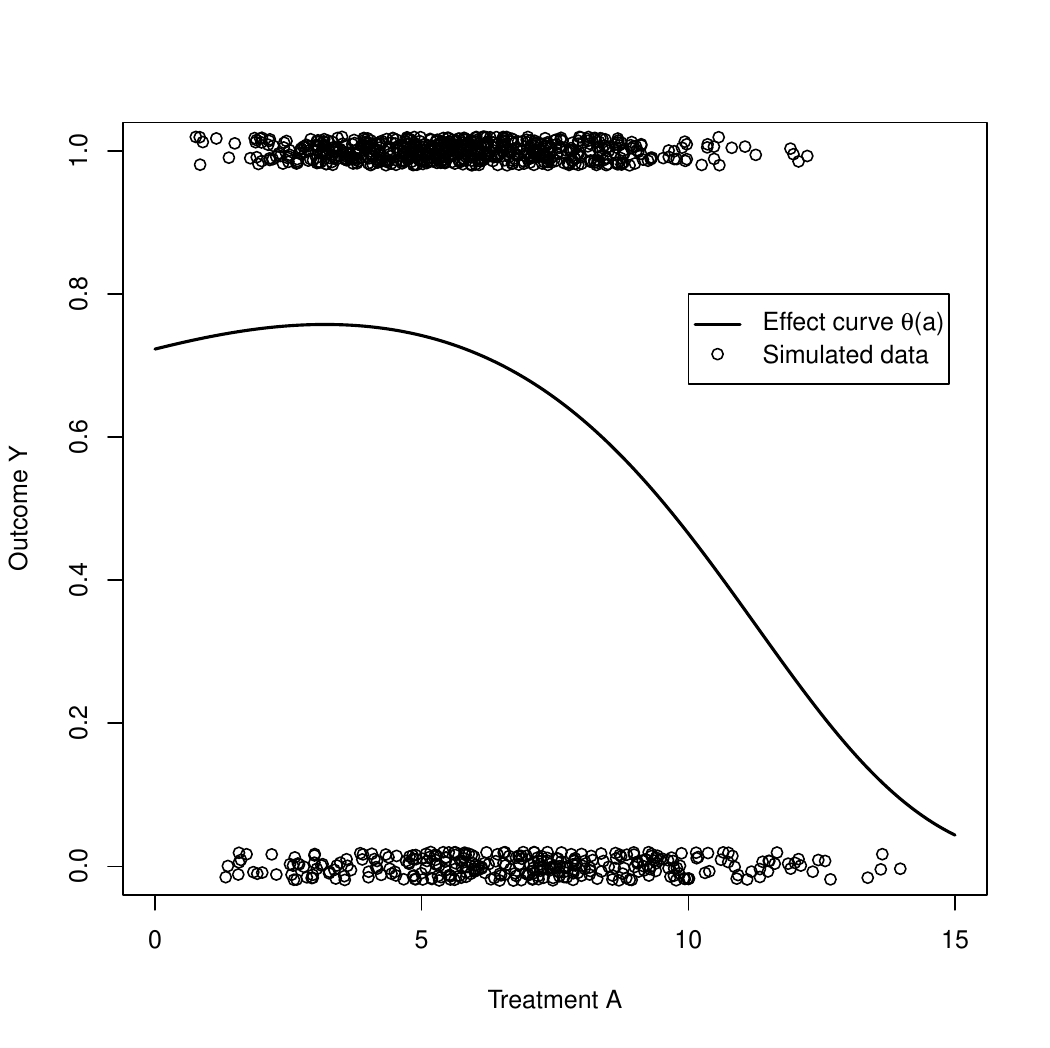}
\caption{Plot of effect curve induced by simulation setup, with treatment and outcome data from one simulated dataset with $n=1000$.}
\label{fig:simplot}
\end{figure}

To analyze the simulated data we used three different estimators: a marginalized regression (plug-in) estimator $\hat{m}(a)=\Pn\{\hat\mu(\bL,a)\}$, and two different versions of the proposed local linear kernel estimator. Specifically we used an inverse-probability-weighted approach first developed by \citet{rubin2006extending}, which relies solely on a treatment model estimator $\hat\pi$ (equivalent to setting $\hat\mu=0$), and the standard doubly robust version that used both estimators $\hat\pi$ and $\hat\mu$. To model the conditional treatment density $\pi$ we used logistic regression to estimate the parameters of the mean function $\lambda(\bl)$; we separately considered correctly specifying this mean function, and then also misspecifying the mean function by transforming the covariates with the same covariate transformations as in \citet{kang2007demystifying}. To estimate the outcome model $\mu$ we again used logistic regression, considering a correctly specified model and then a misspecified model that used the same transformed covariates as with $\pi$ and also left out the cubic term in $a$ (but kept all other interactions). To select the bandwidth we used the leave-one-out approach proposed in Section 3.5, which treats the pseudo-outcomes as known. For comparison we also considered an oracle approach that picked the bandwidth by minimizing the oracle risk $\Pn[\{\theta(A)-\hat\theta_h(A)\}^2]$. In both cases we found the minimum bandwidth value in the range $\mathcal{H} = [0.01, 50]$ using numerical optimization.

We generated 500 simulated datasets for each of three sample sizes, $n=100$, $1000$, and $10000$. To assess the quality of the estimates across simulations we calculated empirical bias and root mean squared error at each point, and integrated across the function with respect to the density of $A$. Specifically, letting $\hat\theta_s(a)$ denote the estimated curve at point $a$ in simulation $s$ ($s=1,...,S$ with $S=500$), we estimated the integrated absolute mean bias and root mean squared error with 
\begin{equation*}
\begin{gathered}
\widehat{\text{Bias}} = \int_{\mathcal{A}^*} \Big| \frac{1}{S} \sum_{s=1}^S \hat\theta_s(a)-\theta(a) \Big| \varpi(a) \ da , \\
\widehat{\text{RMSE}} = \int_{\mathcal{A}^*} \left[ \frac{1}{S} \sum_{s=1}^S \{ \hat\theta_s(a) - \theta(a) \}^2 \right]^{1/2} \!\! \varpi(a) \ da . 
\end{gathered}
\end{equation*}
In the above $\mathcal{A}^*$ denotes a trimmed version of the support of $A$, excluding 10\% of mass at the boundaries. Note that the above integrands (except for the density) correspond to the usual definitions of absolute mean bias and root mean squared error for estimation of a single scalar parameter (e.g., the curve at a single point).

\begin{table}
\caption{\label{tab02}Integrated bias and root mean squared error (500 simulations)}
\centering
\fbox{ %
\begin{tabular}{ll r r r r}
& & \multicolumn{4}{c}{\em Bias (RMSE) when correct model is:} \\
\em$n$ &\em Method & \em Neither & \em Treatment & \em Outcome & \em Both \\
\hline
$100$ 
	& Reg  & 2.67 (5.54) 		& 2.67 (5.54) 	& 0.62 (5.25) 	& 0.62 (5.25) \\
	& IPW & 2.26 (8.49) 		& 1.64 (8.57) 	& 2.26 (8.49) 	& 1.64 (8.57) \\ 
	& IPW* & 2.26 (7.36) 	& 1.58 (7.37) 	& 2.26 (7.36) 	& 1.58 (7.37) \\
	& DR   & 2.23 (6.27) 		& 1.01 (6.28) 	& 1.12 (5.92) 	& 1.10 (6.50) \\
	& DR* & 2.12 (5.48)		& 1.00 (5.36) 	& 1.03 (5.08) 	& 1.02 (5.65) \\ 
& & \\
$1 000$ 
	& Reg  & 2.62 (3.07) 		& 2.62 (3.07) 	& 0.06 (1.53) 	& 0.06 (1.53) \\
	& IPW & 2.38 (3.97) 		& 0.86 (2.94) 	& 2.38 (3.97) 	& 0.86 (2.94) \\ 
	& IPW* & 2.11 (3.44) 	& 0.70 (2.34) 	& 2.11 (3.44) 	& 0.70 (2.34) \\ 
	& DR   & 2.03 (3.11) 		& 0.75 (2.39) 	& 0.74 (2.53) 	& 0.68 (2.25) \\
	& DR* & 1.84 (2.67) 		& 0.64 (1.88) 	& 0.61 (1.78) 	& 0.58 (1.78) \\ 
& & \\
$10 000$ 
	& Reg  & 2.65 (2.70) 		& 2.65 (2.70) 	& 0.02 (0.47) 	& 0.02 (0.47) \\
	& IPW & 2.36 (3.42) 		& 0.33 (1.09) 	& 2.36 (3.42) 	& 0.33 (1.09) \\ 
	& IPW* & 2.24 (3.28) 	& 0.35 (0.85) 	& 2.24 (3.28) 	& 0.35 (0.85) \\ 
	& DR   & 1.81 (2.35) 		& 0.26 (0.86) 	& 0.20 (1.21) 	& 0.25 (0.78) \\
	& DR* & 1.76 (2.27) 		& 0.31 (0.68) 	& 0.24 (1.10) 	& 0.29 (0.64) \\ 
\hline
Notes: & \multicolumn{5}{p{.5\textwidth}}{Bias / RMSE = integrated mean bias / root mean squared error;  IPW = inverse probability weighted; Reg = regression; DR = doubly robust; * = uses oracle bandwidth.}\\
\end{tabular}}
\end{table}

The simulation results are given in Table 1 (both the integrated bias and root mean squared error are multiplied by 100 for easier interpretation). Estimators with stars (e.g., IPW*) denote those with bandwidths selected using the oracle risk. When both $\hat\pi$ and $\hat\mu$ were misspecified, all estimators gave substantial integrated bias and mean squared error (although the doubly robust estimator consistently performed better than the other estimators in this setting). Similarly, all estimators had relatively large mean squared error in the small sample size setting ($n=100$) due to lack of precision, although differences in bias were similar to those at moderate and small sample sizes ($n=1000,10000$). Specifically the regression estimator gave small bias when $\hat\mu$ was correct and large bias when $\hat\mu$ was misspecified, while the inverse-probability-weighted estimator gave small bias when $\hat\pi$ was correct and large bias when $\hat\pi$ was misspecified. However, the doubly robust estimator gave small bias as long as either $\hat\pi$ or $\hat\mu$ was correctly specified, even if one was misspecified. 

The inverse-probability-weighted estimator was least precise, although it had smaller mean squared error than the misspecified regression estimator for moderate and large sample sizes. The doubly robust estimator was roughly similar to the inverse-probability-weighted estimator when the treatment model was correct, but gave less bias and was more precise, and was similar to the regression estimator when the outcome model was correct (but slightly more biased and less precise). In general the estimators based on the oracle-selected bandwidth were similar to those using the simple leave-one-out approach, but gave marginally less bias and mean squared error for small and moderate sample sizes. The benefits of the oracle bandwidth were relatively diminished with larger sample sizes.

\section{Application}

In this section we apply the proposed methodology to estimate the effect of nurse staffing on hospital readmissions penalties, as discussed in the Introduction. In the original paper, \citet{mchugh2013hospitals} used a matching approach to control for hospital differences, and found that hospitals with more nurse staffing were less likely to be penalized; this suggests increasing nurse staffing to help curb excess readmissions. However, their analysis only considered the effect of higher nurse staffing versus lower nurse staffing, and did not explore the full effect curve; it also relied solely on matching for covariate adjustment, i.e., was not doubly robust. 

In this analysis we use the proposed kernel smoothing approach to estimate the full effect curve flexibly, while also allowing for doubly robust covariate adjustment. We use the same data on 2976 acute care hospitals as in \citet{mchugh2013hospitals}; full details are given in the original paper. The covariates $\bL$ include hospital size, teaching intensity, not-for-profit status, urban versus rural location, patient race proportions, proportion of patients on Medicaid, average socioeconomic status, operating margins, a measure of market competition, and whether open heart or organ transplant surgery is performed. The treatment $A$ is nurse staffing hours, measured as the ratio of registered nurse hours to adjusted patient days, and the outcome $Y$ indicates whether the hospital was penalized due to excess readmissions. Excess readmissions are calculated by the Centers for Medicare \& Medicaid Services and aim to adjust for the fact that different hospitals see different patient populations. Without unmeasured confounding, the quantity $\theta(a)$ represents the proportion of hospitals that would have been penalized had all hospitals changed their nurse staffing hours to level $a$. Otherwise $\theta(a)$ can be viewed as an adjusted measure of the relationship between nurse staffing and readmissions penalties.

For the conditional density $\pi(a \mid \bl)$ we assumed a model $A=\lambda(\bL) + \gamma(\bL)\varepsilon$, where $\varepsilon$ has mean zero and unit variance given the covariates, but otherwise has an unspecified density. We flexibly estimated the conditional mean function $\lambda(\bl)=\E(A \mid \bL=\bl)$ and variance function $\gamma(\bl)=\var(A \mid \bL=\bl)$ by combining linear regression, multivariate adaptive regression splines, generalized additive models, Lasso, and boosting, using the cross-validation-based Super Learner algorithm \citep{van2007super}, in order to reduce chances of model misspecification. A standard kernel approach was used to estimate the density of $\varepsilon$. 

For the outcome regression $\mu(\bl,a)$ we again used the Super Learner approach, combining logistic regression, multivariate adaptive regression splines, generalized additive models, Lasso, and boosting. To select the bandwidth parameter $h$ we used the leave-one-out approach discussed in Section 3.5. We considered regression, inverse-probability-weighted, and doubly robust estimators, as in the simulation study in Section 4. The two hospitals ($<$0.1\%) with smallest inverse-probability weights were removed as outliers. For the doubly robust estimator we also computed pointwise uncertainty intervals (i.e., confidence intervals around the smoothed parameter $\theta^*_h(a)$; see Section 3.4) using a Wald approach based on the empirical variance of the estimating function values. 

\begin{figure}[h!]
\centering
\includegraphics[width=0.8\textwidth]{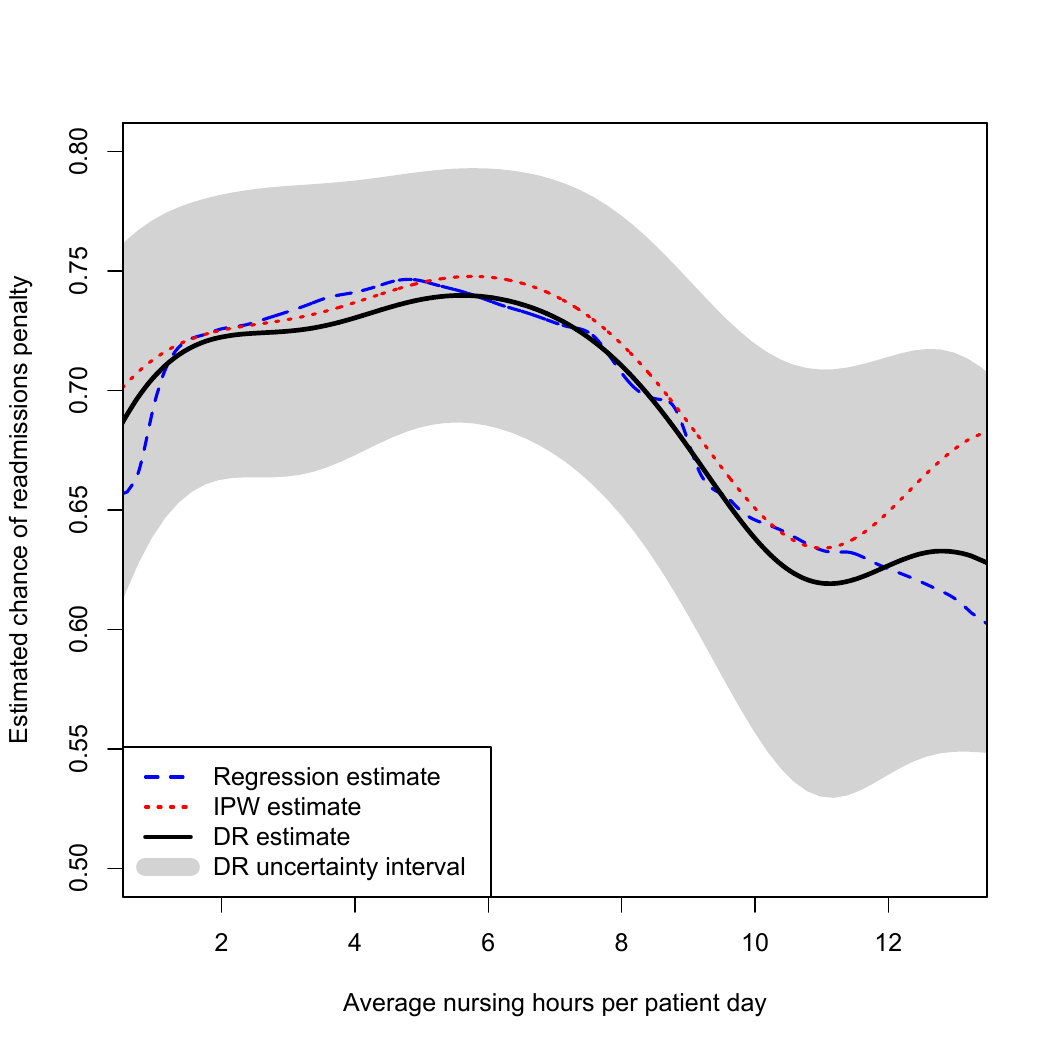}
\caption{Estimated effects of nurse staffing on readmissions penalties.}
\label{fig:resplot}
\end{figure}

A plot showing the results from the three estimators (with uncertainty intervals for the proposed doubly robust estimator) is given in Figure \ref{fig:resplot}. In general the three estimators were very similar. For less than 5 average nurse staffing hours the adjusted chance of penalization was estimated to be roughly constant, around 73\%, but at 5 hours chances of penalization began decreasing, reaching approximately 60\% when nurse staffing reached 11 hours. This suggests that adding nurse staffing hours may be particularly beneficial in the 5-10 hour range, in terms of reducing risk of readmissions penalization; most hospitals (65\%) lie in this range in our data.

\section{Discussion}

In this paper we developed a novel approach for estimating the average effect of a continuous treatment; importantly the approach allows for flexible doubly robust covariate adjustment without requiring any parametric assumptions about the form of the effect curve, and can incorporate general machine learning and nonparametric regression methods. We presented a novel efficient influence function for a stochastic intervention parameter defined within a nonparametric model; this influence function motivated the proposed approach, but may also be useful to estimate on its own. In addition we provided asymptotic results (including rates of convergence and asymptotic normality) for a particular kernel estimation version of our method, which only require the effect curve to be twice continuously differentiable, and allows for flexible data-adaptive estimation of nuisance functions. These results show the double robustness of the proposed approach, since either a conditional treatment density or outcome regression model can be misspecified and the proposed estimator will still be consistent as long as one such nuisance function is correctly specified. We also showed how double robustness can result in smaller second-order bias even when both nuisance functions are consistently estimated. Finally, we proposed a simple and fast data-driven cross-validation approach for bandwidth selection, found favorable finite sample properties of our proposed approach in a simulation study, and applied the kernel estimator to examine the effects of hospital nurse staffing on excess readmissions penalty.

This paper integrates semiparametric (doubly robust) causal inference with nonparametric function estimation and machine learning, helping to bridge the ``huge gap between classical semiparametric models and the model in which nothing is assumed'' \citep{van2014higher}. In particular our work extends standard nonparametric regression by allowing for complex covariate adjustment and doubly robust estimation, and extends standard doubly robust causal inference methods by allowing for nonparametric smoothing. Many interesting problems arise in this gap between standard nonparametric and semiparametric inference, leading to many opportunities for important future work, especially for complex non-regular target parameters that are not pathwise differentiable. In the context of this paper, in future work it will be useful to study uniform distributional properties of our proposed estimator (e.g., weak convergence), as well as its role in testing and inference (e.g., for constructing tests that have power to detect a wide array of deviations from the null hypothesis of no effect of a continuous treatment).

\section{Acknowledgements}

Edward Kennedy was supported by NIH grant R01-DK090385, Zongming Ma by NSF CAREER award DMS-1352060, and Dylan Small by NSF grant SES-1260782. The authors thank Marshall Joffe and Alexander Luedtke for helpful discussions, and two referees for very insightful comments and suggestions. 

\bibliography{bibliography}
\bibliographystyle{rss}


\pagebreak

\begin{figure}[h!]
\centering
\includegraphics[width=0.8\textwidth]{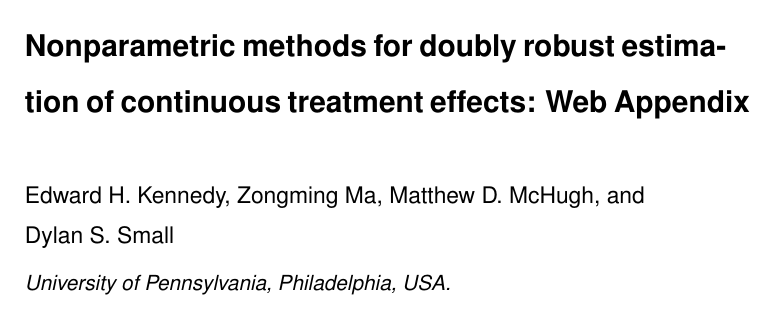}
\end{figure}

\setcounter{page}{1}
\thispagestyle{empty}
\setcounter{section}{0}

\section{Guide to notation}

\noindent $\bZ=(\bL,A,Y)$ = observed data arising from distribution $P$ with density  $p(\bz)=p(y \mid \bl,a) p(a \mid \bl) p(\bl)$ and support $\text{supp}(\bZ)=\mathcal{Z} = \mathcal{L} \times \mathcal{A} \times \mathcal{Y}$ \\

\noindent $\Pn = \frac{1}{n} \sum_i \delta_{\bZ_i}$ = empirical measure so that $\Pn(f) = \Pn\{ f(\bZ)\} = \frac{1}{n} \sum_i f(\bz_i)$ \\

\noindent $\Pb(f) = \Pb\{f(\bZ)\} = \int_\mathcal{Z} f(\bz) \ dP(\bz)$ = expectation for new $\bZ$ treating $f$ as fixed (so $\Pb(\hat{f})$ is random if $\hat{f}$ depends on sample, in which case $\Pb(\hat{f}) \neq \E(\hat{f})$) \\

\noindent $\pi(a \mid \bl) = p(a \mid \bl) = \frac{\partial}{\partial a} P(A \leq a \mid \bl)$ = conditional density of treatment $A$ \\

\noindent $\hat\pi(a \mid \bl)$ = user-specified estimator of $\pi(a \mid \bl)$, which converges to limit $\overline\pi(a \mid \bl)$ that may not equal true $\pi$ \\

\noindent $\varpi(a) = p(a) = \frac{\partial}{\partial a} P(A \leq a) = \E\{\pi(a \mid \bL)\} = \int_\mathcal{L} \pi(a \mid \bl) \ dP(\bl)$ = density of $A$ \\

\noindent $\hat\varpi(a) = \Pn\{ \hat\pi(a \mid \bL) \} = \int_\mathcal{L} \hat\pi(a \mid \bl) \ d\Pn(\bl) = \frac{1}{n} \sum_i \hat\pi(a \mid \bl_i)$ = estimator of $\varpi$, which converges to limit $\overline\varpi(a)$ that may not equal true $\varpi$ \\

\noindent $\mu(\bl,a) = \E(Y \mid \bL=\bl, A=a) = \int_\mathcal{Y} y \ dP(y \mid \bl,a)$ = conditional mean outcome \\

\noindent $\hat\mu(\bl,a)$ = user-specified estimator of $\mu(\bl,a)$, which converges to limit $\overline\mu(\bl,a)$ that may not equal true $\mu$ \\

\noindent $\hat{m}(a) = \Pn\{\hat\mu(\bL,a)\} = \int_\mathcal{L} \hat\mu(\bl,a) \ d\Pn(\bl) = \frac{1}{n} \sum_i \hat\mu(\bl_i,a)$ = regression-based plug-in estimator of $\theta(a)$, which converges to limit $\overline{m}(a)$ that may not equal true $\theta$

\section{Proof of Theorem 1}

Let $p(\bz;\epsilon)$ be a parametric submodel with parameter $\epsilon \in \R$ and $p(\bz;0)=p(\bz)$, for example $p(\bz;\epsilon) = \{1+\epsilon b(\bz)\} p(\bz)$ where $\E\{b(\bZ)\}=0$ with $|b(\bZ)| < B$ and $|\epsilon| \leq (1/B)$ to ensure that $p(\bz;\epsilon)\geq0$. For notational simplicity we denote $\{\partial f(\mathbf{t};\epsilon)/\partial \epsilon\}|_{\epsilon=0}$ by $f'_\epsilon(\mathbf{t};0)$ for any general function $f$ of $\epsilon$ and other arguments $\mathbf{t}$. 

By definition the efficient influence function for $\psi$ is the unique function $\phi(\bZ)$ that satisfies $\psi'_\epsilon(0)= \E\{ \phi(\bZ) \ell'_\epsilon(\bZ;0) \}$, where $\psi(\epsilon)$ represents the parameter of interest as a functional on the parametric submodel and $\ell(\mathbf{w} \mid \mathbf{\overline{w}}; \epsilon)=\log p(\mathbf{w} \mid \mathbf{\overline{w}};\epsilon)$ for any partition $(\mathbf{W},\overline{\mathbf{W}}) \subseteq \bZ$. Therefore
$$ \ell'_\epsilon(\bz;\epsilon) = \ell'_\epsilon(y \mid \bl,a;\epsilon) + \ell'_\epsilon(a \mid \bl;\epsilon) + \ell'_\epsilon(\bl;\epsilon) .$$

We give two important properties of such score functions $\ell'_\epsilon(\mathbf{w} \mid \mathbf{\overline{w}}; \epsilon)$ that will be used throughout this proof. First note that since $\ell(\mathbf{w} \mid \mathbf{\overline{w}}; \epsilon)$ is a log transformation of $p(\mathbf{w} \mid \mathbf{\overline{w}};\epsilon)$, it follows that $\ell'_\epsilon(\mathbf{w} \mid \mathbf{\overline{w}}; \epsilon)=p'_\epsilon(\mathbf{w} \mid \mathbf{\overline{w}};\epsilon) /p(\mathbf{w} \mid \mathbf{\overline{w}};\epsilon)$ because for general functions $f$ we have $\partial \log f(\epsilon)/\partial \epsilon=\{\partial f(\epsilon)/\partial \epsilon\}/f(\epsilon)$. Similarly, as with any score function, note that $\E\{\ell'_\epsilon(\mathbf{W} \mid \mathbf{\overline{W}}; 0) \mid \mathbf{\overline{W}}\}=0$ since
\begin{align*}
\int_\mathbf{\mathcal{W}} \ell'_\epsilon(\mathbf{w} \mid \mathbf{\overline{w}}; 0) \ dP(\mathbf{w} \mid \mathbf{\overline{w}}) = \int_\mathbf{\mathcal{W}} dP'_\epsilon(\mathbf{w} \mid \mathbf{\overline{w}}) = \frac{\partial}{\partial \epsilon} \int_\mathbf{\mathcal{W}} dP(\mathbf{w} \mid \mathbf{\overline{w}}) = 0 .
\end{align*}

Our goal in this proof is to show that $\psi'_\epsilon(0)= \E\{ \phi(\bZ) \ell'_\epsilon(\bZ;0) \}$ for the proposed influence function $\phi(\bZ)=\xi(\bZ;\pi,\mu)-\psi+\int_\mathcal{A} \{ \mu(\bL,a) - \int_\mathcal{L} \mu(\bl,a) dP(\bl) \} \varpi(a) da$ given in the main text. First we will give an expression for $\psi'_\epsilon(0)$. By definition $ \psi(\epsilon) = \int_\mathcal{A} \theta(a; \epsilon) \varpi(a; \epsilon) \ da $, so
\begin{align*}
\psi'_\epsilon(0) &= \int_\mathcal{A} \{ \theta'_\epsilon(a; 0) \varpi(a) + \theta(a) \varpi'_\epsilon(a;0) \} \ da  
= \E\{ \theta'_\epsilon(A; 0)  + \theta(A) \ell'_\epsilon(A;0) \} .
\end{align*}
Also since $\theta(a;\epsilon) = \int_\mathcal{L} \int_\mathcal{Y} y \ p(y \mid \bl,a; \epsilon) p(\bl;\epsilon) \ d\eta(y) \ d\nu(\bl)$, we have
\begin{align*}
\theta'_\epsilon(a; 0) &= \int_\mathcal{L} \int_\mathcal{Y} y \Big\{ p'_\epsilon(y \mid \bl,a; 0) p(\bl) + p(y \mid \bl,a) p'_\epsilon(\bl;0) \Big\} d\eta(y) \ d\nu(\bl) \\
&= \int_\mathcal{L} \int_\mathcal{Y} y \Big\{ \ell'_\epsilon(y \mid \bl,a; 0) p(y \mid \bl,a) p(\bl) + p(y \mid \bl,a) \ell'_\epsilon(\bl;0) p(\bl) \Big\} d\eta(y) \ d\nu(\bl) \\
&=\E\Big[ \E\{ Y \ell'_\epsilon(Y \mid \bL,A; 0) \mid \bL, A=a \} \Big] + \E\Big\{ \mu(\bL,a) \ell'_\epsilon(\bL;0)\Big\} .
\end{align*}
Therefore
\begin{align*}
\psi'_\epsilon(0)  = \int_\mathcal{A} \bigg( \E\Big[ &\E\{ Y \ell'_\epsilon(Y \mid \bL,A; 0) \mid \bL, A=a \} \Big] \\
&+ \E\Big\{ \mu(\bL,a) \ell'_\epsilon(\bL;0)\Big\}  + \theta(a) \ell'_\epsilon(a;0) \bigg) \varpi(a) \ da .
\end{align*}

Now we will consider the covariance 
$$\E\{ \phi(\bZ) \ell'_\epsilon(\bZ;0) \}=\E\Big[ \phi(\bZ) \Big\{ \ell'_\epsilon(Y \mid \bL,A;0) + \ell'_\epsilon(A ,\bL;0) \Big\} \Big],$$
which we need to show equals the earlier expression for $\psi'_\epsilon(0)$.

Recall the proposed efficient influence function given in the main text is
\begin{align*}
\frac{Y - \mu(\bL,A)}{\pi(A \mid \bL)} \varpi(A) + m(A) - \psi + \int_\mathcal{A} \Big\{ \mu(\bL,a) - m(a) \Big\} \varpi(a) \ da
\end{align*}
where we define
$$ m(a) = \int_\mathcal{L} \mu(\bl,a) \ dP(\bl) $$
as the marginalized version of the regression function $\mu$, so that $m(a)=\theta(a)$ if $\mu$ is the true regression function.

Thus $\E\{ \phi(\bZ) \ell'_\epsilon(Y \mid \bL,A;0) \}$ equals
\begin{align*}
\E\bigg( \bigg[ &\frac{Y - \mu(\bL,A)}{\pi(A \mid \bL)/ \varpi(A)} + \int_\mathcal{A} \Big\{ \mu(\bL,a) - \theta(a) \Big\} \varpi(a) \ da + \theta(A) - \psi  \bigg] \ell'_\epsilon(Y \mid \bL,A;0)\bigg) \\
&= \E\bigg\{ \frac{Y \ell'_\epsilon(Y \mid \bL,A;0)}{\pi(A \mid \bL)/\varpi(A)} \bigg\} = \E\bigg[ \frac{\E\{Y \ell'_\epsilon(Y \mid \bL,A;0) \mid \bL,A\}}{\pi(A \mid \bL)/\varpi(A)} \bigg] \\
&= \int_\mathcal{A} \E\Big[\E\{Y \ell'_\epsilon(Y \mid \bL,A;0) \mid \bL,A=a\}\Big] \varpi(a) \ da
\end{align*}
where the first equality follows since $\E\{ \ell'_\epsilon(Y \mid \bL,A;0) \mid \bL, A\}=0$, the second by iterated expectation conditioning on $\bL$ and $A$, and the third by iterated expectation conditioning on $\bL$. Now note that $\E\{ \phi(\bZ) \ell'_\epsilon(A ,\bL;0)\}$ equals
\begin{align*}
&\E \bigg[ \bigg\{ \frac{Y - \mu(\bL,A)}{\pi(A \mid \bL)/\varpi(A)} \bigg\} \ell'_\epsilon(A ,\bL;0) + \{ \theta(A) - \psi \} \Big\{ \ell'_\epsilon(\bL \mid A;0) + \ell'_\epsilon(A;0) \Big\} \\
& \hspace{.25in} + \int_\mathcal{A}\Big\{ \mu(\bL,a)-\theta(a)\Big\} \varpi(a) \ da \ \Big\{ \ell'_\epsilon(A \mid \mathbf{L};0) + \ell'_\epsilon(\bL;0) \Big\} \bigg] \\
&=  \E \bigg[ \theta(A) \ell'_\epsilon(A;0) + \int_\mathcal{A} \mu(\bL,a) \ell'_\epsilon(\bL;0) \varpi(a) \ da  \bigg] 
\end{align*}
since by definition $\ell'_\epsilon(A ,\bL;0)=\ell'_\epsilon(A \mid \bL;0)+\ell'_\epsilon(\bL;0)=\ell'_\epsilon(\bL \mid A;0)+\ell'_\epsilon(A;0)$, and the equality used iterated expectation conditioning on $\bL$ and $A$ for the first term in the first line, $A$ for the second term in the first line, and $\bL$ for the second line. Adding the expressions $\E\{ \phi(\bZ) \ell'_\epsilon(Y \mid \bL,A;0) \}$ and $\E\{ \phi(\bZ) \ell'_\epsilon(A ,\bL;0)\}$ gives
$$  \int_\mathcal{A} \bigg( \E\Big[ \E\{ Y \ell'_\epsilon(Y \mid \bL,A; 0) \mid \bL, A=a \} +  \mu(\bL,a) \ell'_\epsilon(\bL;0) \Big]  + \theta(a) \ell'_\epsilon(a;0) \bigg) \varpi(a) \ da  , $$
which equals $\psi'_\epsilon(0)$. Thus $\phi$ is the efficient influence function.

\section{Double robustness of efficient influence function \& mapping}

Here we will show that $\E\{\phi(\bZ;\overline\pi,\overline\mu,\psi)\} = 0$ if either $\overline\pi=\pi$ or $\overline\mu=\mu$, where $\phi(\bZ;\overline\pi,\overline\mu,\psi)$ is the influence function defined as in the main text as
$$\xi(\bZ;\overline\pi,\overline\mu) - \psi  + \int_\mathcal{A} \Big\{ \overline\mu(\bL,a)- \int_\mathcal{L} \overline\mu(\bl,a) \ dP(\bl)\Big\} \int_\mathcal{L} \overline\pi(a \mid \bl) \ dP(\bl) \ da, $$
where
$$ \xi(\bZ;\overline\pi,\overline\mu) = \frac{Y - \overline\mu(\bL,A)}{\overline\pi(A \mid \bL)} \int_\mathcal{L} \overline\pi(A \mid \bl) \ dP(\bl)  +  \int_\mathcal{L} \overline\mu(\bl,A) \ dP(\bl) . $$

First note that, letting $\overline\varpi(a)=\E\{\overline\pi(a \mid \bL)\}$ and $\overline{m}(a)=\E\{\overline\mu(\bL,a)\}$, we have
\begin{align*}
\E\{\xi(\bZ;\overline\pi,\overline\mu) &\mid A=a\} = \E\left\{ \frac{Y - \overline\mu(\bL,A)}{\overline\pi(A \mid \bL)/\overline\varpi(A)}  +  \overline{m}(A) \Bigm| A=a \right\} \\
&= \int_{\mathcal{L}} \frac{\mu(\bl,a) - \overline\mu(\bl,a)}{\overline\pi(a \mid \bl)/\overline\varpi(a)}   \ dP(\bl \mid a)  + \overline{m}(a) \\
&= \int_{\mathcal{L}} \Big\{ \mu(\bl,a) - \overline\mu(\bl,a) \Big\} \frac{\pi(a \mid \bl)/\varpi(a)}{\overline\pi(a \mid \bl)/\overline\varpi(a)}   \ dP(\bl)  + \overline{m}(a) \\
&= \theta(a)  + \int_{\mathcal{L}} \Big\{ \mu(\bl,a) - \overline\mu(\bl,a) \Big\} \left\{ \frac{\pi(a \mid \bl)/\varpi(a)}{\overline\pi(a \mid \bl)/\overline\varpi(a)} -1 \right\}  \ dP(\bl) 
\end{align*}
where the first equality follows by iterated expectation, the second follows since $p(\bl \mid a) = p(a \mid \bl) p(\bl) / p(a)$, and the third by rearranging. The last line shows that $\E\{\xi(\bZ;\overline\pi,\overline\mu) \mid A=a\}=\theta(a)$ as long as either $\overline\pi=\pi$ or $\overline\mu=\mu$, since in either case the remainder is zero.

Therefore if $\overline\pi=\pi$ or $\overline\mu=\mu$ we have
$$  \int_{\mathcal{A}} \E\{ \xi(\bZ;\overline\pi,\overline\mu) \mid A=a \} \varpi(a) \ da - \psi = \int_\mathcal{A} \theta(a) \varpi(a) \ da - \psi = 0 $$
so that
\begin{align*}
\E\{ \phi(\bZ;\overline\pi,\overline\mu,\psi)\} =  \E \left[ \int_\mathcal{A} \Big\{ \overline\mu(\bL,a)-\overline{m}(a) \Big\} \overline\varpi(a) \ da \right] .
\end{align*}

But
\begin{align*}
\E \int_\mathcal{A} \Big\{ \overline\mu(\bL,a)-\overline{m}(a) \Big\} \overline\varpi(a) \ da = \int_\mathcal{A} \Big\{ \overline{m}(a)-\overline{m}(a) \Big\} \overline\varpi(a) \ da = 0
\end{align*}
by definition.

Therefore $\E\{\phi(\bZ;\overline\pi,\overline\mu,\psi)\} = 0$ if either $\overline\pi=\pi$ or $\overline\mu=\mu$.

\section{TMLE version of estimator}

As we note in the main text, the proposed estimator 
$$ \hat\theta_h(a) = \bg_{ha}(a)^\T \Pn\{ \bg_{ha}(A) K_{ha}(A) \bg_{ha}(A)^\T\}^{-1} \Pn\{ \bg_{ha} K_{ha}(A) \hat\xi(\bZ;\hat\pi,\hat\mu)\} $$
is not guaranteed to respect bounds on $Y$, e.g., if $Y \in [0,1]$ is binary. If some observations have very small values of the denominator quantity $\hat\pi(A \mid \bL)/\hat\varpi(A)$ then the estimator could be unstable and may take values outside the range of $Y$. Targeted maximum likelihood or minimum loss-based estimators (TMLEs), developed by \citet{van2006targeted}, help combat this problem (see discussion for example in \citet{van2011targeted} and elsewhere). In this section we present a TMLE that should give better finite-sample performance, for example, when there are near-violations of the positivity assumption.

Our proposed TMLE can be fit as follows. First estimate the nuisance functions $\hat\pi$ and $\hat\mu$, for example with flexible machine learning (e.g., Super Learner). Then fit a logistic regression model regressing $Y$ on `clever covariate' vector
$$ \mathbf{\hat{c}}_{ha}(\bL,A)= \frac{\bg_{ha}(A) K_{ha}(A)}{\hat\pi(A \mid \bL) / \hat\varpi(A)} $$
with $\text{logit}\{\hat\mu(\bL,A)\}$ included as an offset (and no intercept term). This ensures
$$ \Pn\bigg\{ \frac{\bg_{ha}(A) K_{ha}(A)}{\hat\pi(A \mid \bL)/\hat\varpi(A)} \bigg( Y - \text{expit} \Big[ \text{logit}\{\hat\mu(\bL,A)\} + \boldsymbol{\hat\epsilon}^\T \mathbf{\hat{c}}_{ha}(\bL,A) \Big] \bigg) \bigg\} = 0 $$
where $\boldsymbol\epsilon=(\epsilon_1,\epsilon_2)$ are the parameters in the logistic regression fit. Now define
$$ \hat\mu_{ha}^*(\bL,A) = \text{expit} \Big[ \text{logit}\{\hat\mu(\bL,A)\} + \boldsymbol{\hat\epsilon}^\T \mathbf{\hat{c}}_{ha}(\bL,A) \Big] . $$
Then the proposed method proceeds as before, simply replacing predicted values $\hat\mu(\bL,A)$ with $\hat\mu_{ha}^*(\bL,A)$. Specifically we estimate $\theta(a)$ with
$$ \hat\theta^*_h(a) = \bg_{ha}(a)^\T \Pn\{ \bg_{ha}(A) K_{ha}(A) \bg_{ha}(A)^\T\}^{-1} \Pn\{ \bg_{ha} K_{ha}(A) \hat{m}_{ha}^*(A) \} , $$
$$ \hat{m}_{ha}^*(t) = \Pn \{ \hat\mu_{ha}^*(\bL,t) \} = \Pn \Big( \text{expit} \Big[ \text{logit}\{\hat\mu(\bL,t)\} + \boldsymbol{\hat\epsilon}^\T \mathbf{\hat{c}}_{ha}(\bL,t) \Big] \Big) . $$
The above TMLE is somewhat more complicated to fit than the estimator proposed in the main text. An alternative approach that would also respect bounds on $Y$ would be to estimate $\theta(a)$ with $\hat\theta_h(a)=\text{expit}\{\bg_{ha}(a)^\T \boldsymbol{\hat\beta}_h(a) \}$ where
$$ \boldsymbol{\hat\beta}_h(a) = \argmin_{\boldsymbol\beta \in \R^2} \ \Pn\left( K_{ha}(A) \Big[ \hat\xi(\bZ;\hat\pi,\hat\mu)  - \text{\expit}\{ \bg_{ha}(A)^\T \boldsymbol\beta \} \Big]^2 \right) . $$
Another simple option would be to use the original estimator from the main text and project onto the range of possible $Y$ values.

\section{Stochastic equicontinuity lemmas}

In this section we discuss the concept of asymptotic or stochastic equicontinuity, and give two lemmas that play a central role in subsequent proofs.

Let $\Gn=\sqrt{n}(\Pn-\Pb)$. A sequence of empirical processes $\{ \Gn V_n(f) : f \in \mathcal{F}\}$ indexed by elements $f$ ranging over a metric space $\mathcal{F}$ (equipped with semimetric $\rho$) is stochastically equicontinuous \citep{pollard1984convergence,andrews1994empirical,van1996weak} if for every $\varepsilon>0$ and $\zeta>0$ there exists a $\delta>0$ such that
$$ \limsup_{n \rightarrow \infty} P\bigg( \sup_{\rho(f_1,f_2)<\delta} | \Gn V_n(f_1) - \Gn V_n(f_2) | > \varepsilon \bigg) < \zeta . $$
An important consequence of stochastic equicontinuity for our purposes is that if  $\{\Gn V_n(\cdot) : n \geq 1 \}$ is stochastically equicontinuous then $\rho(\hat{f},\overline{f}) = o_p(1)$ implies that $\Gn\{ V_n(\hat{f}) - V_n(\overline{f})\} = o_p(1)$ \citep{pollard1984convergence,andrews1994empirical}.

Before presenting relevant lemmas, we first need to introduce some notation. Let $F$ denote an envelope function for the space $\mathcal{F}$, i.e., a function with $F(\bz)\geq  |f(\bz) |$ for every $f \in \mathcal{F}$ and $\bz \in \mathcal{Z}$. Also let $N(\varepsilon,\mathcal{F},||\cdot||)$ denote the covering number, i.e., the minimal number of $\varepsilon$-balls (using distance $||\cdot||$) needed to cover $\mathcal{F}$, and let 
$$J(\delta,\mathcal{F},L_2)=\int_0^\delta \sup_Q  \sqrt{\log N(\varepsilon ||F||_{Q,2}, \mathcal{F},L_2(Q))} \ d\varepsilon , $$ 
where $L_2(Q)$ denotes the usual $L_2$ semimetric under distribution $Q$, which for any $f$ is $|| f ||_{Q,2}=(\int f^2 dQ)^{1/2}$. We call $J(\infty,\mathcal{F},L_2)$ the uniform entropy integral. 

To show that a sequence of processes $\{\Gn V_n(\cdot) : n \geq 1 \}$ as defined above is stochastically equicontinuous, one can use Theorem 2.11.1 from \citet{van1996weak}. (Note that in their notation $Z_n(f)=(1/\sqrt{n}) V_n(f)$.) Specifically, Theorem 2.11.1 states that stochastic equicontinuity follows from the following two Lindeberg conditions (conditions 1 and 2), with an additional restriction on the complexity of the space $\mathcal{F}$ (condition 3):
\begin{enumerate}
\item[(1)] $\E\{ ||V_n||_\mathcal{F}^2 \ I(||V_n||_\mathcal{F} > \varepsilon \sqrt{n})\} \rightarrow 0$ for every $\varepsilon>0$.
\item[(2)] $\sup_{\rho(f_1,f_2)<\delta_n} \E[ \{V_n(f_1)-V_n(f_2)\}^2] \rightarrow 0$ for every sequence $\delta_n \rightarrow 0$.
\item[(3)] $\int_0^{\delta_n} \sqrt{\log N(\varepsilon, \mathcal{F}, L_2(\Pn))} \ d\varepsilon \inprob 0$ for every sequence $\delta_n \rightarrow 0$.
\end{enumerate}

We will give conditions under which two particular kinds of sequences of empirical processes are stochastically equicontinuous. Specifically we consider processes $\{\Gn V_n(\cdot) : n \geq 1 \}$ where
\begin{align*}
V_n(f) &= \sqrt{h} \ g_{ha}(A) K_{ha}(A) f(\bZ) , \\
V_n(f) &= \int f(\bL,t) g_{ha}(t) K_{ha}(t) \ dt ,  
\end{align*}
with $g_{ha}(t)$ and $K_{ha}(t)$ defined earlier (note $V_n$ depends on $n$ since $h=h_n$ does).

\begin{lemma}
Consider the sequence of processes $\{\Gn V_{n,j}(\cdot) : n \geq 1 \}$ with
$$ V_{n,j}(f) = \sqrt{h} \left(\frac{A-a}{h}\right)^{j-1} \frac{1}{h} K\!\left(\frac{A-a}{h}\right) f(\bZ) \ , \ \ j=1,2 , $$ 
where $f \in \mathcal{F}$ with envelope $F(\bz)=\sup_{f \in \mathcal{F}} |f(\bz)|$. Assume the following:
\begin{enumerate}
\item The bandwidth $h=h_n$ satisfies $h \rightarrow 0$ and $nh^3 \rightarrow \infty$ as $n \rightarrow \infty$.
\item The kernel $K$ is a bounded symmetric probability density with support $[-1,1]$.
\item $A$ has compact support $\mathcal{A}$ and continuous density $\varpi$.
\item The envelope $F$ is uniformly bounded, i.e., $||F||_\mathcal{Z} \leq f_{max} < \infty$.
\item $\mathcal{F}$ has a finite uniform entropy integral, i.e., $J(\delta,\mathcal{F},L_2)<\infty$.
\end{enumerate}
Then $\{\Gn V_{n,j}(\cdot) : n \geq 1 \}$ is stochastically equicontinuous.
\end{lemma}

\begin{proof}

Recall that to show stochastic equicontinuity we can check conditions (1)--(3) of Theorem 2.11.1 from \citet{van1996weak}, as given earlier. 

We will show Lindeberg condition (1) using the dominated convergence theorem, which says if $X_n \inprob X$ and $|X_n| \leq Y$ with $\E(Y) < \infty$ then $\E(X_n) \rightarrow \E(X)$. First note that $||V_{n,j}||_\mathcal{F}^2 \ I(||V_{n,j}||_\mathcal{F} > \varepsilon \sqrt{n}) = o_p(1)$ since for any $\delta>0$
\begin{align*}
\lim_{n \rightarrow \infty}  P&\Big\{ ||V_{n,j}||_\mathcal{F}^2 \ I(||V_{n,j}||_\mathcal{F} > \varepsilon \sqrt{n}) \geq \delta \Big\} \\
&\leq \lim_{n \rightarrow \infty}  P\Big( ||V_{n,j}||_\mathcal{F} > \varepsilon \sqrt{n} \Big) \\
&= \lim_{n \rightarrow \infty}  P\Big\{ (A-a)^{j-1} K\left(\frac{A-a}{h}\right) F(Z) > \varepsilon \sqrt{nh^{2j-1}} \Big\} \\
&\leq \lim_{n \rightarrow \infty}  P\Big\{ (A-a)^{j-1} ||K||_{[-1,1]} f_{max} > \varepsilon \sqrt{nh^{2j-1}} \Big\}  .
\end{align*}
The last line above used the kernel and envelope conditions (b) and (c). The expression in the last line tends to zero as $n \rightarrow \infty$, since $nh \rightarrow \infty$ and $nh^3 \rightarrow \infty$ by the bandwidth condition (a) (note that $nh \rightarrow \infty$ is implied by the fact that $h \rightarrow 0$ and $nh^3 \rightarrow \infty$), and since $A$ has compact support by condition (c). We also have $||V_{n,j}||_\mathcal{F}^2 I(||V_{n,j}||_\mathcal{F} > \varepsilon \sqrt{n}) \leq ||V_{n,j}||_\mathcal{F}^2$ since $I(\cdot)$ is the indicator function, and $\E\{ ||V_{n,j}||_\mathcal{F}^2 \} < \infty$ since
\begin{align*}
\E\{ ||V_{n,j}||_\mathcal{F}^2 \} &= \E\Big[ \left(\frac{A-a}{h}\right)^{2(j-1)} \frac{1}{h} K\!\left(\frac{A-a}{h}\right)^2 F(Z)^2 \Big] \\
&\leq f_{max}^2 || \varpi ||_\mathcal{A} \ \int \left(\frac{t-a}{h}\right)^{2(j-1)} \frac{1}{h} K\!\left(\frac{A-a}{h}\right)^2  \ dt \\
&= f_{max}^2 || \varpi ||_\mathcal{A} \int u^{2(j-1)} K(u)^2  \ dt < \infty .
\end{align*}
The second line above follows by the distribution condition (c) and the envelope condition (d), and the last line is finite by the kernel properties assumed in condition (b). Therefore since $||V_{n,j}||_\mathcal{F}^2 \ I(||V_{n,j}||_\mathcal{F} > \varepsilon \sqrt{n}) = o_p(1)$ and $||V_{n,j}||_\mathcal{F}^2 \ I(||V_{n,j}||_\mathcal{F} > \varepsilon \sqrt{n}) \leq ||V_{n,j}||_\mathcal{F}^2$ with $\E\{ ||V_{n,j}||_\mathcal{F}^2 \} < \infty$, the dominated convergence theorem implies that $\E\{ ||V_{n,j}||_\mathcal{F}^2 \ I(||V_{n,j}||_\mathcal{F} > \varepsilon \sqrt{n})\} \rightarrow 0$ as $n \rightarrow \infty$ and thus Lindeberg condition (1) holds.

Lindeberg condition (2) holds when $\rho(\cdot)$ is the uniform norm since
\begin{align*}
\sup_{\rho(f_1,f_2)< \delta_n } &\E[ \{V_{n,j}(f_1)-V_{n,j}(f_2)\}^2] \\
&= \sup_{||f_1-f_2||_{\mathcal{Z}}< \delta_n } \E\left[ \left(\frac{A-a}{h}\right)^{2(j-1)} \frac{1}{h} K\left(\frac{A-a}{h}\right)^2 \Big\{ f_1(\bZ) - f_2(\bZ) \Big\}^2 \right] \\
&\leq \delta_n^2 \ \int \left(\frac{t-a}{h}\right)^{2(j-1)} \frac{1}{h} K\left(\frac{t-a}{h}\right)^2 \varpi(t) \ dt \\
&\leq \delta_n^2 \ || \varpi ||_\mathcal{A}  \int u^{2(j-1)} K(u)^2 \ dt \ \rightarrow \ 0 \ \ , \ \ \text{for any $\delta_n \rightarrow 0$}.
\end{align*}
The first equality above follows by definition, the second inequality by the fact that $||f_1-f_2||_{\mathcal{Z}} <\delta_n$, and the third by condition (c) and a change of variables. The last line tends to zero as $\delta_n \rightarrow 0$ by the kernel properties in condition (b).

Now we consider the complexity condition (3). As described in Section 2.11.1.1 (page 209) of \citet{van1996weak}, a process $(1/\sqrt{n}) V_n(f)$ is measure-like if for some (random) measure $\nu_{ni}$ we have
$$ \frac{1}{n} \Big\{ V_n(f_1)-V_n(f_2) \Big\}^2 \leq \int (f_1 - f_2)^2 \ d\nu_{ni}  \ , \ \ \text{for every $f_1, f_2 \in \mathcal{F}$} . $$
\citet{van1996weak} show in their Lemma 2.11.6 that if $\mathcal{F}$ has a finite uniform entropy integral, then measure-like processes indexed by $\mathcal{F}$ satisfy the complexity condition (3) of Theorem 2.11.1. 

Note that for our process $V_{n,j}(f)$ of interest, we have
\begin{align*}
\frac{1}{n} \Big\{ V_{n,j}(f_1)-V_{n,j}(f_2) \Big\}^2 &= \Big\{ f_1(\bZ) - f_2(\bZ) \Big\}^2 \sqrt{h} \left(\frac{A-a}{h}\right)^{j-1} \frac{1}{h} K\!\left(\frac{A-a}{h}\right) .
\end{align*}
Therefore the processes $V_{n,j}(f)$ are measure-like for the random measure $\nu_{ni}=\sqrt{h} g_{ha} K_{ha} \delta_{\bZ_i}$, where $\delta_{\bZ_i}$ denotes the Dirac measure. Hence, by Lemma 2.11.6 of \citet{van1996weak}, the fact that $\mathcal{F}$ has a finite uniform entropy integral (assumed in condition (e)) implies that complexity condition (3) is satisfied. 

Therefore the sequence $\{\Gn V_{n,j}(\cdot) : n \geq 1 \}$ is stochastically equicontinuous. 

\end{proof}

As mentioned earlier, Lemma 1 implies that if $|| \hat{f}-f||_\mathcal{Z} = o_p(1)$ then 
$$ \sqrt{nh} (\Pn-\Pb) \left[ \left(\frac{A-a}{h}\right)^{j-1} \frac{1}{h} K\!\left(\frac{A-a}{h}\right) \Big\{ \hat{f}(\bZ) - f(\bZ) \Big\}\right] =  o_p(1). $$

\begin{lemma}
Consider the sequence of processes $\{\Gn V_{n,j}(\cdot) : n \geq 1 \}$ with
$$ V_{n,j}(f) = \int f(\bL,t) \Big(\frac{t-a}{h}\Big)^{j-1} \frac{1}{h} K\!\Big(\frac{t-a}{h}\Big) \ dt \ , \ \ j=1,2 , $$ 
where $f \in \mathcal{F}$ with envelope $F$ as in Lemma 1. Assume conditions (b), (d), and (e) of Lemma 1 hold. Then $\{\Gn V_{n,j}(\cdot) : n \geq 1 \}$ is stochastically equicontinuous.
\end{lemma}

\begin{proof}

The proof of Lemma 2 is very similar to that of Lemma 1. We again show Lindeberg condition (1) using the dominated convergence theorem. First note $||V_{n,j}||_\mathcal{F}^2 \ I(||V_{n,j}||_\mathcal{F} > \varepsilon \sqrt{n}) = o_p(1)$ since for any $\delta>0$
\begin{align*}
\lim_{n \rightarrow \infty}  P&\Big\{ ||V_{n,j}||_\mathcal{F}^2 \ I(||V_{n,j}||_\mathcal{F} > \varepsilon \sqrt{n}) \geq \delta \Big\} \leq \lim_{n \rightarrow \infty}  P\Big( ||V_{n,j}||_\mathcal{F} > \varepsilon \sqrt{n} \Big) \\
&= \lim_{n \rightarrow \infty}  P\Big\{ \int F(\bL,t) \{(t-a)/h\}^{j-1} K\{(t-a)/h\}/h \ dt > \varepsilon \sqrt{n} \Big\} \\
&\leq \lim_{n \rightarrow \infty}  I\Big\{ f_{max} \int |u|^{j-1} K(u) \ dt > \varepsilon \sqrt{n} \Big\} = 0.
\end{align*}
The last line above used the kernel and envelope conditions (b) and (d). We also have $||V_{n,j}||_\mathcal{F}^2 I(||V_{n,j}||_\mathcal{F} > \varepsilon \sqrt{n}) \leq ||V_{n,j}||_\mathcal{F}^2$ and $\E\{||V_{n,j}||_\mathcal{F}^2\}$ equals
$$ \bigg\{ \int F(\bL,t) \left(\frac{t-a}{h}\right)^{j-1} \frac{1}{h} K\!\left(\frac{t-a}{h}\right) dt \bigg\}^2 \\
\leq f_{max}^2 \bigg\{ \int |u|^{j-1} K(u) \ du \bigg\}^2,  $$ 
which is finite again using conditions (b) and (d). Therefore Lindeberg condition (1) holds since $\E\{||V_{n,j}||_\mathcal{F}^2 I(||V_{n,j}||_\mathcal{F} > \varepsilon \sqrt{n})\} \rightarrow 0$ by dominated convergence.

Lindeberg condition (2) holds with the uniform norm since, by definition and using the kernel condition (b),  $\sup_{\rho(f_1,f_2)<\delta_n} \E[ \{V_n(f_1)-V_n(f_2)\}^2 ]$ equals
\begin{align*}
\sup_{||f_1-f_2||_\mathcal{Z} < \delta_n } \! \E\bigg( \bigg[ &\int \Big\{ f_1(\bL,t) - f_2(\bL,t) \Big\} \left(\frac{t-a}{h}\right)^{j-1} \frac{1}{h} K\!\left(\frac{t-a}{h}\right)  dt \bigg]^2\bigg) \\
&\leq \delta_n^2 \bigg\{ \int |u|^{j-1} K(u) \ du \bigg\}^2 \rightarrow 0 \ \ , \ \ \text{for any $\delta_n \rightarrow 0$}.
\end{align*}

As in Lemma 1, we use that $V_{n,j}$ is measure-like to check condition (3). Here
\begin{align*}
\frac{1}{n} \{ V_{n,j}(f_1)&-V_{n,j}(f_2) \}^2 = \frac{1}{n} \bigg[ \int \{ f_1(\bL,t)-f_2(\bL,t)\} \Big(\frac{t-a}{h}\Big)^{j-1} \frac{1}{h} K\!\Big(\frac{t-a}{h}\Big) dt \bigg]^2 \\
&\leq \frac{1}{n}  \int \Big\{ f_1(\bL,t)-f_2(\bL,t)\Big\}^2 \left|\frac{t-a}{h}\right|^{2(j-1)} \frac{1}{h} K\!\left(\frac{t-a}{h}\right) dt 
\end{align*}
by Jensen's inequality. Therefore the processes $V_{n,j}(f)$ are measure-like, and the fact that $\mathcal{F}$ has a finite uniform entropy integral (assumed in condition (e)) implies that complexity condition (3) is satisfied. This concludes the proof.
\end{proof}

\section{Proof of Theorem 2}

Here we let $\tilde\theta_h(a)=\bg_{ha}(a)^\T \mathbf{\hat{D}}_{ha}^{-1} \Pn\{ \bg_{ha}(A) K_{ha}(A) \xi(\bZ; \overline\pi,\overline\mu)\}$ denote the infeasible estimator one would use if the nuisance functions were known, with $\mathbf{\hat{D}}_{ha}=\Pn\{\bg_{ha}(A) K_{ha}(A) \bg_{ha}(A)^\T \}$ as in the main text. Our proposed estimator is $\hat\theta_h(a)=\bg_{ha}(a)^\T \mathbf{\hat{D}}_{ha}^{-1} \Pn\{\bg_{ha}(A) K_{ha}(A) \hat\xi(\bZ;\hat\pi,\hat\mu)\}$. We use the decomposition
$$ \hat\theta_h(a) - \theta(a) =  \Big\{ \tilde\theta_h(a)-\theta(a) \Big\} + \Big\{ \hat\theta_h(a) - \tilde\theta_h(a) \Big\} = \Big\{ \tilde\theta_h(a) - \theta(a) \Big\} + (R_{n,1} + R_{n,2}) $$
where 
\begin{align*}
R_{n,1} &= \bg_{ha}(a)^\T \mathbf{\hat{D}}_{ha}^{-1} (\Pn-\Pb) \left[ \bg_{ha}(A) K_{ha}(A) \Big\{ \hat\xi(\bZ;\hat\pi,\hat\mu) - \xi(\bZ;\overline\pi,\overline\mu) \Big\} \right] \\
R_{n,2} &= \bg_{ha}(a)^\T \mathbf{\hat{D}}_{ha}^{-1} \Pb \left[ \bg_{ha}(A) K_{ha}(A) \Big\{ \hat\xi(\bZ;\hat\pi,\hat\mu) - \xi(\bZ;\overline\pi,\overline\mu) \Big\} \right] .
\end{align*}
 Our proof is divided into three parts, one for the analysis of each of the terms above.

\subsection{Convergence rate of $\tilde\theta_h(a)-\theta(a)$}

Since the infeasible estimator $\tilde\theta_h(a)$ is a standard local linear kernel estimator with outcome $\xi(\bZ;\overline\pi,\overline\mu)$ and regressor $A$, it can be analyzed with results from the local polynomial kernel regression literature. In particular, since our Assumption 2 (Positivity) along with conditions (b), (c), (d) of our Theorem 2 imply the bandwidth condition and conditions 1(i)-1(iv) in \citet{fan1993local}, by their Theorem 1 we have
$\E[ \tilde\theta_h(a) - \E\{ \xi(\bZ;\overline\pi,\overline\mu) \mid A=a\} ]^2 = O( 1/nh + h^4 )$. 
Further, condition (a) of our Theorem 1 implies $\E\{ \xi(\bZ;\overline\pi,\overline\mu) \mid A=a\}=\theta(a)$ by the results in Section 3 of this Appendix. Therefore $\E\{ \tilde\theta_h(a) -\theta(a)\}^2 = O( 1/nh + h^4 )$.

Now let $X_n=\tilde\theta_h(a)-\theta(a)$. The above implies that, for some $M^*>0$, $\limsup_{n \rightarrow \infty} \ \E\{ X_n^2 / (1/nh+h^4) \} \leq M^*$. Therefore for any $\epsilon > 0$, if $M \geq M^*/\epsilon$,
$$ \lim_{n \rightarrow \infty} P\left( \frac{X_n^2}{1/nh + h^4} \geq M \right) \leq \limsup_{n \rightarrow \infty} \frac{1}{M} \E\left( \frac{X_n^2}{1/nh+ h^4} \right) \leq M^*/M \leq \epsilon $$
where the first equality follows by Markov's inequality, the second by the fact that $\E(X_n^2) = O( 1/nh + h^4 )$, and the third by definition of $M$. Since $\epsilon>0$ was arbitrary this implies $\{ \tilde\theta_h(a) -\theta(a)\}^2 = O_p( 1/nh + h^4 )$.

Now let $b_n=1/\sqrt{nh} + h^2$ and $c_n = 1/nh + h^4$, and note that 
$$ P\left( \left| \frac{X_n}{b_n} \right| \geq \sqrt{M} \right) = P\left( \left| \frac{X_n^2}{c_n + 2h\sqrt{h/n}} \right| \geq M \right) \leq P\left( \left| \frac{X_n^2}{c_n} \right| \geq M \right) . $$
Taking limits as $n \rightarrow \infty$ implies that
$$ \Big| \tilde\theta_h(a) - \theta(a) \Big| = O_p\left( \frac{1}{\sqrt{nh}} + h^2 \right) . $$

\subsection{Asymptotic negligibility of $R_{n,1}$}

Now we will show that
$$ R_{n,1} = \bg_{ha}(a)^\T \mathbf{\hat{D}}_{ha}^{-1} (\Pn-\Pb) \left[ \bg_{ha}(A) K_{ha}(A) \Big\{ \hat\xi(\bZ;\hat\pi,\hat\mu) - \xi(\bZ;\overline\pi,\overline\mu) \Big\} \right] $$
is asymptotically negligible up to order $\sqrt{nh}$, i.e., $|R_{n,1}|=o_p(1/\sqrt{nh})$. 

First we will show that $\bg_{ha}(a)^\T \mathbf{\hat{D}}_{ha}^{-1}=O_p(1)$. Consider the elements of the matrix $\mathbf{\hat{D}}_{ha}$. Using the continuity of $\varpi$ from condition (d) of Theorem 2 in the main text, along with properties of the kernel function from condition (c), it is straightforward to show that
$$ \E\Big( [ \Pn\{K_{ha}(A)\} - \varpi(a) ]^2 \Big) = O(h) + O(1/nh) . $$
Hence $\E( [ \Pn\{K_{ha}(A)\} - \varpi(a) ]^2)=o(1)$, since $h \rightarrow 0$ and $nh \rightarrow \infty$ by condition (b), and therefore $\Pn\{K_{ha}(A)\} \inprob \varpi(a)$ by Markov's inequality. This is a standard result in classical kernel estimation problems. By the same logic we similarly have
\begin{equation*}
\begin{gathered}
\Pn\{ K_{ha}(A) (A-a)/h \} \ \inprob \ 0 , \\
\Pn[ K_{ha}(A) \{(A-a)/h\}^2 ] \ \inprob \ \varpi(a) \int u^2 K(u) \ du . 
\end{gathered}
\end{equation*}
Therefore  $ \bg_{ha}(a)^\T \mathbf{\hat{D}}_{ha}^{-1} \inprob \begin{pmatrix} 1 & 0 \end{pmatrix} \text{diag}\{ \varpi(a), \varpi(a) \nu_2 \}^{-1} = \begin{pmatrix} \varpi(a)^{-1} & 0 \end{pmatrix} , $
where $\text{diag}(c_1,c_2)$ is a $(2\times 2)$ diagonal matrix with elements $c_1$ and $c_2$ on the diagonal, $\nu_2=\int u^2 K(u) \ du$, and $\varpi(a) \neq 0$ because of Assumption 2 (Positivity). Thus we have shown that $\bg_{ha}(a)^\T \mathbf{\hat{D}}_{ha}^{-1}=\begin{pmatrix} \varpi(a)^{-1} & 0 \end{pmatrix} + o_p(1)=O_p(1)$.

Now we will analyze the term 
$$ (\Pn-\Pb) \left[ \bg_{ha}(A) K_{ha}(A) \Big\{ \hat\xi(\bZ;\hat\pi,\hat\mu) - \xi(\bZ;\overline\pi,\overline\mu) \Big\} \right] , $$
which we will show is $o_p(1/\sqrt{nh})$. This is equivalent to showing
$$ \Gn \left[ \sqrt{h} \ \bg_{ha}(A) K_{ha}(A) \hat\xi(\bZ) \right] = \Gn \left[ \sqrt{h} \ \bg_{ha}(A) K_{ha}(A) \overline\xi(\bZ) \Big\} \right] + o_p(1) , $$
where we define $\hat\xi(\bZ)=\hat\xi(\bZ;\hat\pi,\hat\mu)$ and $\overline\xi(\bZ)=\xi(\bZ;\overline\pi,\overline\mu)$. Note that, as discussed in the previous section on stochastic equicontinuity, if $||\hat\xi-\overline\xi||_\mathcal{Z}= o_p(1)$ then the above result follows if the sequence of empirical processes $\{ \Gn V_n( \cdot) : n\geq 1\}$ is stochastically equicontinuous, where we define $V_n(\xi)=\sqrt{h} \bg_{ha}(A) K_{ha}(A) \xi(\bZ)$ with $\xi \in \Xi$ for some metric space $\Xi$. Thus first we will show that $||\hat\xi-\overline\xi||_\mathcal{Z}=\sup_{\bz \in \mathcal{Z}} |\hat\xi(\bz;\hat\pi,\hat\mu)-\xi(\bz;\overline\pi,\overline\mu)| = o_p(1)$. Then we will check the conditions given in Lemma 1 of the previous section, which ensure that $\{ \Gn V_n( \cdot) : n\geq 1\}$ defined above is stochastically equicontinuous.

First note that after some rearranging we can write
\begin{align*}
\hat\xi(\bz) &- \xi(\bz) = \frac{y-\hat\mu(\bl,a)}{\hat\pi(a \mid \bl)} \hat\varpi(a) + \hat{m}(a) - \frac{y-\overline\mu(\bl,a)}{\overline\pi(a \mid \bl)} \overline\varpi(a) - \overline{m}(a) \\
&= \frac{y-\overline\mu(\bl,a)}{\overline\pi(a \mid \bl)}  \frac{\hat\varpi(a)}{\hat\pi(a \mid \bl)} \Big\{ \overline\pi(a \mid \bl) - \hat\pi(a \mid \bl) \Big\} + \frac{\hat\varpi(a)}{\hat\pi(a \mid \bl)}  \Big\{ \overline\mu(\bl,a) - \hat\mu(\bl,a) \Big\} \\
& \hspace{.4in} + \frac{y-\overline\mu(\bl,a)}{\overline\pi(a \mid \bl)} \Big\{ \hat\varpi(a)-\overline\varpi(a)\Big\} + \Big\{ \hat{m}(a) - \overline{m}(a) \Big\} .
\end{align*}
Therefore, letting $\hat\xi(\bz) = \xi(\bz;\hat\pi,\hat\mu)$ and similarly $\overline\xi(\bz)=\xi(\bz;\overline\pi,\overline\mu)$, by the uniform boundedness assumed in condition (e) and the triangle inequality we have
$$ || \hat\xi - \overline\xi ||_\mathcal{Z} = O_p\Big( ||\hat\pi-\overline\pi||_\mathcal{Z} + ||\hat\mu-\overline\mu||_\mathcal{Z} + || \hat\varpi-\overline\varpi||_\mathcal{A}  + ||\hat{m}-\overline{m}||_\mathcal{A}  \Big) . $$
Therefore since $||\hat\pi-\overline\pi||_\mathcal{Z}=o_p(1)$ and $||\hat\mu-\overline\mu||_\mathcal{Z} = o_p(1)$ by definition, and since $O_p(o_p(1))=o_p(1)$, the above implies
$$ || \hat\xi - \overline\xi ||_\mathcal{Z} = O_p\Big( || \hat\varpi-\overline\varpi||_\mathcal{A}  + ||\hat{m}-\overline{m}||_\mathcal{A}  \Big) + o_p(1) . $$
Now, since by definition $\hat\varpi(a)=\Pn\{\hat\pi(a \mid \bL)\}$ and $\overline\varpi(a)=\E\{\overline\pi(a \mid \bL)\}$, we have that
\begin{align*}
|| \hat\varpi - \overline\varpi ||_\mathcal{A} &= \sup_{a \in \mathcal{A}} | \hat\varpi(a) - \overline\varpi(a) | = \sup_{a \in \mathcal{A}} \Big| \Pn \hat\pi(a \mid \bL) - \Pb \overline\pi(a \mid \bL) \Big| \\
&= \sup_{a \in \mathcal{A}} \Big| \Pn \{\hat\pi(a \mid \bL) - \overline\pi(a \mid \bL)\}+ (\Pn- \Pb) \overline\pi(a \mid \bL) \Big| \\
&\leq \sup_{a \in \mathcal{A}} \Big| \Pn \{\hat\pi(a \mid \bL) - \overline\pi(a \mid \bL)\} \Big| + \sup_{a \in \mathcal{A}} \Big| (\Pn- \Pb) \overline\pi(a \mid \bL) \Big| \\
&\leq || \hat\pi - \overline\pi||_\mathcal{Z} + \sup_{a \in \mathcal{A}} \Big| (\Pn- \Pb) \overline\pi(a \mid \bL) \Big| , 
\end{align*}
where the last two lines used the triangle inequality. By definition the first term on the right hand side of the last line is $o_p(1)$, and by the uniform entropy assumption in condition (e) the second term is also $o_p(1)$ since it implies that $\overline\pi$ is Glivenko-Cantelli \citep{van2000asymptotic,van1996weak}. Therefore we have $||\hat\varpi-\overline\varpi||_\mathcal{Z}=o_p(1)$. By exactly the same logic, using definitions and condition (e) we similarly have
\begin{align*}
|| \hat{m} - \overline{m} ||_\mathcal{A} 
&\leq \sup_{a \in \mathcal{A}} \Big| \Pn \{\hat\mu(\bL,a) - \overline\mu( \bL,a) \} \Big| + \sup_{a \in \mathcal{A}} \Big| (\Pn- \Pb) \overline\mu( \bL,a) \Big| \\
&\leq || \hat\mu - \overline\mu||_\mathcal{Z} + \sup_{a \in \mathcal{A}} \Big| (\Pn- \Pb) \overline\mu( \bL,a) \Big| = o_p(1) .
\end{align*}
Therefore $||\hat\xi-\overline\xi||_\mathcal{Z}=\sup_{\bz \in \mathcal{Z}} |\hat\xi(\bz;\hat\pi,\hat\mu)-\xi(\bz;\overline\pi,\overline\mu)| = o_p(1)$. 

Now we will show that the conditions given in Lemma 1 hold, indicating that the sequence $\{ \Gn V_n( \cdot) : n\geq 1\}$ defined above is stochastically equicontinuous. Conditions (a)--(c) of Lemma 1 are given exactly in the statement of Theorem 2 and so hold immediately. For conditions (d) and (e) of Lemma 1 we need to consider the space $\Xi$ containing elements $\xi(\bz)$. The space $\Xi$ can be constructed as a transformation of the spaces $(\mathcal{F}_\pi,\mathcal{F}_\mu,\mathcal{F}_\varpi,\mathcal{F}_m)$ containing the functions $(\pi,\mu,\varpi,m)$, along with the single identity function that takes $\bZ$ as input and outputs $Y$. Specifically, we have
$$ \Xi = (Y \oplus \mathcal{F}_\mu) \mathcal{F}_\pi^{-1} \mathcal{F}_\varpi \oplus \mathcal{F}_m $$
where $Y$ is shorthand for the single function that outputs $Y$ from $\bZ$, and we define 
$\mathcal{F}_1 \oplus \mathcal{F}_2 = \{ f_1 + f_2 : f_j \in \mathcal{F}_j\}$, $\mathcal{F}^{-1} = \{ 1/f : f \in \mathcal{F}\}$, and similarly $\mathcal{F}_1 \mathcal{F}_2 = \{ f_1 f_2 : f_j \in \mathcal{F}_j\}$,
for arbitrary function classes $\mathcal{F}$ and $\mathcal{F}_j$ containing functions $f$ and $f_j$ respectively. For more discussion of such constructions of higher-level function spaces based on lower-level building blocks, we refer the reader to \citet{pollard1990empirical} (Section 5), \citet{andrews1994empirical} (Section 4.1), \citet{van1996weak} (Section 2.10), and \citet{van2000asymptotic} (Examples 19.18--19.20); for use in a related example and more discussion see \citet{van2006estimating} (Section 5).

By condition (e) of Theorem 2, the classes $(\mathcal{F}_\pi,\mathcal{F}_\mu,\mathcal{F}_\varpi,\mathcal{F}_m)$ are uniformly bounded (i.e., their minimal envelopes are bounded above by some constant). Similarly the class $\mathcal{F}_\pi^{-1}$ is also uniformly bounded by the second part of condition (e). Therefore the constructed class $\Xi$ is bounded as well, so that condition (d) of Lemma 1 holds.

Condition (e) of Lemma 1 can be verified by using permanence or stability properties of the uniform entropy integral \citep{andrews1994empirical,van1996weak,van2006estimating}. Specifically, by condition (e) of Theorem 2,  the classes $(\mathcal{F}_\pi,\mathcal{F}_\mu,\mathcal{F}_\varpi,\mathcal{F}_m)$ all have a finite uniform entropy integral (as does the single function $Y$, or any finite set of functions). Therefore by Theorem 3 of \citet{andrews1994empirical}, since $\mathcal{F}_\pi^{-1}$ is appropriately bounded with finite envelope, it follows that the class $\Xi$ also has a finite uniform entropy integral. Thus condition (e) of Lemma 1 holds. For results similar to Theorem 3 of \citet{andrews1994empirical}, also see Theorem 2.10.20 of \citet{van1996weak}, and Lemma 5.1 and subsequent examples of \citet{van2006estimating}.

Thus since the conditions of Lemma 1 hold, the sequence $\{ \Gn V_n( \cdot) : n\geq 1\}$ with $V_n(\xi)=\sqrt{h} \bg_{ha}(A) K_{ha}(A) \xi(\bZ)$ is stochastically equicontinuous, and since $||\hat\xi-\overline\xi||_\mathcal{Z}=\sup_{\bz \in \mathcal{Z}} |\hat\xi(\bz;\hat\pi,\hat\mu)-\xi(\bz;\overline\pi,\overline\mu)| = o_p(1)$, it therefore follows that
$$ (\Pn-\Pb) \left[ \bg_{ha}(A) K_{ha}(A) \Big\{ \hat\xi(\bZ;\hat\pi,\hat\mu) - \xi(\bZ;\overline\pi,\overline\mu) \Big\} \right] = o_p(1/\sqrt{nh}) . $$
Combined with the fact that $\bg_{ha}(a)^\T \mathbf{\hat{D}}_{ha}^{-1}=O_p(1)$, this implies that $R_{n,1}=o_p(1/\sqrt{nh})$ and so is asymptotically negligible.

\subsection{Convergence rate of $R_{n,2}$}

In this section we will derive the convergence rate of
$$ R_{n,2} = \bg_{ha}(a)^\T \mathbf{\hat{D}}_{ha}^{-1} \Pb \left[ \bg_{ha}(A) K_{ha}(A) \Big\{ \hat\xi(\bZ;\hat\pi,\hat\mu) - \xi(\bZ;\overline\pi,\overline\mu) \Big\} \right] , $$
which will depend on how well the nuisance functions $\pi$ and $\mu$ are estimated. 

In the previous subsection we showed that $\bg_{ha}(a)^\T \mathbf{\hat{D}}_{ha}^{-1}=O_p(1)$ using conditions (b), (c), and (d) of Theorem 3, along with Assumption 2 (Positivity). Therefore we will consider the term $\Pb [ \bg_{ha}(A) K_{ha}(A) \{ \hat\xi(\bZ;\hat\pi,\hat\mu) - \xi(\bZ;\overline\pi,\overline\mu) \} ]$,
which is a vector with $j^{th}$ element ($j=1,2$) equal to
$$ \int_\mathcal{A} g_{ha,j}(t) K_{ha}(t) \ \Pb\Big\{ \hat\xi(\bZ;\hat\pi,\hat\mu) - \xi(\bZ;\overline\pi,\overline\mu) \mid A=t \Big\} \varpi(t) \ dt , $$
where $g_{ha,j}(t)=\{(t-a)/h\}^{j-1}$ as before. Note that
\begin{align*}
\Pb\{ \hat\xi&(\bZ;\hat\pi,\hat\mu) - \xi(\bZ;\overline\pi,\overline\mu) \mid A=t \} = \Pb\left\{ \frac{Y-\hat\mu(\bL,A)}{\hat\pi(A \mid \bL)/\hat\varpi(A)} \Bigm| A=t \right\} + \hat{m}(t) - \theta(t) \\
&= \Pb\left[ \Big\{\mu(\bL,t)-\hat\mu(\bL,t) \Big\} \left\{ \frac{\pi(t \mid \bL)/\varpi(t)}{\hat\pi(t \mid \bL)/\hat\varpi(t)} \right\} \right] + \hat{m}(t) - \theta(t) \\
&= \frac{\hat\varpi(t)}{\varpi(t)} \ \Pb\!\left[ \Big\{\mu(\bL,t)-\hat\mu(\bL,t) \Big\} \left\{ \frac{\pi(t \mid \bL) - \hat\pi(t \mid \bL)}{\hat\pi(t \mid \bL)} \right\} \right] \\
& \hspace{.5in} + \frac{1}{\varpi(t)} \ \Pb\Big\{\hat\pi(t \mid \bL)-\pi(t \mid \bL) \Big\} \Pb\Big\{\mu(\bL,t)-\hat\mu(\bL,t) \Big\} \\
& \hspace{.5in} + \frac{\Pb\{\mu(\bL,t)-\hat\mu(\bL,t) \}}{\varpi(t)} (\Pn-\Pb) \{\hat\pi(t \mid \bL)\} + (\Pn- \Pb)\{\hat\mu(\bL,t) \} .
\end{align*}
The first equality above follows since $\E\{\xi(\bZ;\overline\pi,\overline\mu) \mid A=t \}=\theta(t)$ because either $\overline\pi=\pi$ or $\overline\mu=\mu$ (as shown in Section 3), the second by iterated expectation and the fact that $p(\bl \mid a)=\{\pi(a \mid \bl)/\varpi(a)\}p(\bl)$, and the third by rearranging terms and the definitions $\hat\varpi(t)=\Pn\{\hat\pi(t \mid \bL)\}$ and $\hat{m}(t)=\Pn\{\hat\mu(\bL,t)\}$. 

Therefore using the Cauchy-Schwarz inequality ($\Pb(fg)\leq ||f|| \ ||g||$), the triangle inequality, Assumption 2 (Positivity), and the uniform boundedness assumed in condition (e), we have
\begin{align*}
\Big| &\Pb \Big[ {g}_{ha,j}(A) K_{ha}(A) \Big\{ \hat\xi(\bZ;\hat\pi,\hat\mu) - \xi(\bZ;\overline\pi,\overline\mu) \Big\} \Big] \Big| \\
&= O_p\bigg( \ \left| \int_\mathcal{A} g_{ha,j}(t) K_{ha}(t) \ ||\hat\pi(t \mid \bL) - \pi(t \mid \bL) || \ || \hat\mu(\bL,t) - \mu(\bL,t) || \ dt \right| \\
& \hspace{.5in} + \left| (\Pn-\Pb) \int_\mathcal{A} g_{ha,j}(t) K_{ha}(t) \ \hat\pi(t \mid \bL) \ dt \right| \\
& \hspace{.5in} +  \left| (\Pn-\Pb) \int_\mathcal{A} g_{ha,j}(t) K_{ha}(t) \ \hat\mu(\bL,t) \ dt \right| \ \bigg).
\end{align*}

The last two terms above can be controlled by Lemma 2 in this Appendix. Specifically, this lemma can be applied since its condition (b) corresponds exactly to condition (b) of Theorem 2, and since its conditions (d) and (e) are implied by condition (e) of Theorem 2. Therefore since $||\hat\pi-\overline\pi||_\mathcal{Z}=o_p(1)$ and $||\hat\mu-\overline\mu||_\mathcal{Z}=o_p(1)$ by definition, the stochastic equicontinuity result of Lemma 2 implies that
\begin{equation*}
\begin{gathered}
 (\Pn-\Pb) \int_\mathcal{A} g_{ha,j}(t) K_{ha}(t) \Big\{ \hat\pi(t \mid \bL) - \overline\pi(t \mid \bL) \Big\} \ dt = o_p(1/\sqrt{n}), 
\end{gathered}
\end{equation*}
and similarly replacing $\pi$ with $\mu$. Therefore by the central limit theorem we have
\begin{equation*}
\begin{gathered}
 (\Pn-\Pb) \int_\mathcal{A} g_{ha,j}(t) K_{ha}(t) \ \hat\pi(t \mid \bL) \ dt = O_p(1/\sqrt{n}) , 
\end{gathered}
\end{equation*}
and similarly replacing $\pi$ with $\mu$.
Thus the last two terms in the inequality on the previous page are asymptotically negligible up to order $\sqrt{nh}$ since 
$$ X_n=O_p(1/\sqrt{n}) \implies \sqrt{n} X_n = O_p(1) \implies \sqrt{nh} X_n = O_p(1) o_p(1) = o_p(1). $$
Therefore since $O_p(o_p(1/\sqrt{nh}))=o_p(1/\sqrt{nh})$, we have
\begin{align*}
\Big| \Pb \Big[ {g}_{ha,j}&(A) K_{ha}(A) \Big\{ \hat\xi(\bZ;\hat\pi,\hat\mu) - \xi(\bZ;\overline\pi,\overline\mu) \Big\} \Big] \Big| \\
&= O_p\bigg( \ \left| \int_\mathcal{A} g_{ha,j}(t) K_{ha}(t) \ \phi_\pi(t) \ \phi_\mu(t) \ dt \right| \ \bigg) + o_p(1/\sqrt{nh}) 
\end{align*}
where $\phi_\pi(t)=||\hat\pi(t \mid \bL) - \pi(t \mid \bL) ||$ and $\phi_\mu(t)= || \hat\mu(\bL,t) - \mu(\bL,t) ||$.

Now let $||K||_{[-1,1]}=K_{max}$. Since $K(u) \leq K_{max} I(|u|\leq 1)$, we have
\begin{align*}
\int_\mathcal{A}  g_{ha,j}&(t) K_{ha}(t) \ \phi_\pi(t) \phi_\mu(t) \ dt = \int_\mathcal{A} \left(\frac{t-a}{h}\right)^{j-1} \frac{1}{h} K\!\left(\frac{t-a}{h}\right)  \phi_\pi(t) \phi_\mu(t) \ dt \\
&\leq K_{max} \left\{ \sup_{t:|t-a| \leq h} \phi_\pi(t) \right\} \left\{ \sup_{t:|t-a| \leq h} \phi_\mu(t) \right\} \int_{-1}^1 |u|^{j-1} \ du .
\end{align*}
In the main text we define $r_n(a)$ and $s_n(a)$ so that $\sup_{t:|t-a| \leq h} \phi_\pi(t)=O_p(r_n(a))$ and $\sup_{t:|t-a| \leq h} \phi_\mu(t)=O_p(s_n(a))$. Therefore
$$ \Big| \Pb \Big[ {g}_{ha,j}(A) K_{ha}(A) \Big\{ \hat\xi(\bZ;\hat\pi,\hat\mu) - \xi(\bZ;\overline\pi,\overline\mu) \Big\} \Big] \Big| = O_p\Big( r_n(a) s_n(a) \Big) . $$
Combining the above with the results from subsections 6.1 and 6.2 yields the desired rate from the statement of Theorem 2,
$$ \Big| \hat\theta_h(a) - \theta(a) \Big| = O_p\left( \frac{1}{\sqrt{nh}} + h^2 + r_n(a) s_n(a) \right)  . $$

\section{Proof of Theorem 3}

As in Theorem 2, we again use the decomposition
$$ \hat\theta_h(a) - \theta(a) =  \Big\{ \tilde\theta_h(a)-\theta(a) \Big\} + \Big\{ \hat\theta_h(a) - \tilde\theta_h(a) \Big\} = \Big\{ \tilde\theta_h(a) - \theta(a) \Big\} + (R_{n,1} + R_{n,2}) $$
where $\tilde\theta_h(a)=\bg_{ha}(a)^\T \mathbf{\hat{D}}_{ha}^{-1} \Pn\{\bg_{ha}(A) K_{ha}(A) \hat\xi(\bZ;\hat\pi,\hat\mu)\}$ is our proposed estimator, $\tilde\theta_h(a)=\bg_{ha}(a)^\T \mathbf{\hat{D}}_{ha}^{-1} \Pn\{\bg_{ha}(A) K_{ha}(A) \xi(\bZ;\overline\pi,\overline\mu)\}$ is the infeasible estimator with known nuisance functions, $\mathbf{\hat{D}}_{ha} = \Pn\{\bg_{ha}(A) K_{ha}(A) \bg_{ha}(A)^\T\}$, and
\begin{align*}
R_{n,1} &= \bg_{ha}(a)^\T \mathbf{\hat{D}}_{ha}^{-1} (\Pn-\Pb) \left[ \bg_{ha}(A) K_{ha}(A) \Big\{ \hat\xi(\bZ;\hat\pi,\hat\mu) - \xi(\bZ;\overline\pi,\overline\mu) \Big\} \right] \\
R_{n,2} &= \bg_{ha}(a)^\T \mathbf{\hat{D}}_{ha}^{-1} \Pb \left[ \bg_{ha}(A) K_{ha}(A) \Big\{ \hat\xi(\bZ;\hat\pi,\hat\mu) - \xi(\bZ;\overline\pi,\overline\mu) \Big\} \right] .
\end{align*}
We consider each term separately, as in the proof of Theorem 2.

\subsection{Asymptotic normality of $\tilde\theta_h(a)-\theta(a)$}

After scaling, the first term $\tilde\theta_h(a)-\theta(a)$ above is asymptotically normal by Theorem 1 from \citet{fan1994robust}, since $\tilde\theta_h(a)$ is a standard local linear kernel estimator with outcome $\xi(\bZ;\overline\pi,\overline\mu)$ and regressor $A$, and since $\E\{\xi(\bZ;\overline\pi,\overline\mu) \mid A=a\}=\theta(a)$ by condition (a) (i.e., either $\overline\pi=\pi$ or $\overline\mu=\mu$) as shown in Section 3 of this Appendix. Similar proofs for the asymptotic normality of local linear kernel estimators can be found elsewhere as well \citep{fan1992design,fan1995local, masry1997local,li2007nonparametric}. Specifically, under conditions (b), (c), and (d) of Theorem 3 stated in the main text, the proof given by \citet{fan1994robust} shows that, for $b_h(a)=\theta''(a)(h^2/2) \int u^2 K(u) \ du$, we have
$$ \sqrt{nh} \Big\{ \tilde\theta_h(a) - \theta(a) - b_h(a) \Big\} \indist N\left(0, \ \frac{\sigma^2(a) \int K(u)^2 \ du }{\varpi(a)} \right) $$
where, using the fact that $\E\{\xi(\bZ;\overline\pi,\overline\mu) \mid A=a\}=\theta(a)$ and rearranging,
\begin{align*}
\sigma^2(a) &\equiv \var\{\xi(\bZ;\overline\pi,\overline\mu)\mid A=a\} \\
&= \E\Big( \Big[ \xi(\bZ;\overline\pi,\overline\mu) - \E\{ \xi(\bZ;\overline\pi,\overline\mu) \mid A=a\} \Big]^2 \Bigm| A=a \Big) \\
&= \E\bigg[ \bigg\{ \frac{Y - \overline\mu(\bL,A)}{\overline\pi(A \mid \bL)/\overline\varpi(A)} + \overline{m}(A) - \theta(A) \bigg\}^2 \Bigm| A=a  \bigg] \\
&= \E\bigg[ \bigg\{ \frac{Y - \overline\mu(\bL,A)}{\overline\pi(A \mid \bL)/\overline\varpi(A)} \bigg\}^2  \Bigm| A=a  \bigg] - \{\theta(a)-\overline{m}(a)\}^2 \\
&= \E\!\left[ \frac{\tau^2(\bL,a) + \{\mu(\bL,a)-\overline\mu(\bL,a)\}^2}{\{\overline\pi(a \mid \bL)/\overline\varpi(a)\}^2 / \{\pi(a \mid \bL)/\varpi(a)\}} \right] - \Big\{ \theta(a)-\overline{m}(a) \Big\}^2 .
\end{align*}

\subsection{Asymptotic negligibility of $R_{n,1}$}

We showed $R_{n,1}=o_p(1/\sqrt{nh})$ in the proof of Theorem 2 in Section 6.2. 

\subsection{Asymptotic negligibility of $R_{n,2}$}

In the proof of Theorem 2 in Section 6.3 of this Appendix, we showed that $R_{n,2} = O_p( r_n(a) s_n(a) )$, 
where $r_n(a)$ and $s_n(a)$ are the local rates of convergence  for the nuisance estimators $\hat\pi$ and $\hat\mu$, as defined in the main text. By condition (f) of Theorem 3, we have $r_n(a) s_n(a) = o_p(1/\sqrt{nh})$ so that $R_{n,2} = O_p( o_p(1/\sqrt{nh}) ) = o_p(1/\sqrt{nh})$, and thus $R_{n,2}$ is asymptotically negligible up to order $\sqrt{nh}$.

Therefore the proposed estimator $\hat\theta_h(a)$ is asymptotically equivalent to the infeasible estimator $\tilde\theta_h(a)$. This yields the result from Theorem 2 in the main text.

\section{Uniform consistency}

In this section we sketch some conditions under which our estimator is not only consistent pointwise but also uniformly in the sense that $ \sup_{a \in \mathcal{A}} |\hat\theta_h(a) - \theta(a) | = o_p(1)$, and give a rate of convergence. However we leave a full treatment of this result to future work, in which we will also explore weak convergence of $\hat\theta_h(a)$ to some Gaussian process. This will be useful for testing and inference.

We use the same decomposition as in Sections 6-7 proving Theorems 2-3,
$$ \hat\theta_h(a) - \theta(a) = \Big\{ \tilde\theta_h(a) - \theta(a) \Big\} + R_{n,1}(a) + R_{n,2}(a) $$
with $R_{n,1}(a)=R_{n,1}$ and $R_{n,2}(a)=R_{n,2}$ defined as before. From \citet{masry1996multivariate} and \citet{hansen2008uniform}  (among others), under standard smoothness/bandwidth conditions we have
$$ \sup_{a \in \mathcal{A}} | \tilde\theta_h(a) - \theta(a) | = O_p\left( \sqrt{\frac{\log n}{nh}} + h^2 \right) . $$
Further, if the empirical process $V_n(a)=\sqrt{nh/\log n} R_{n,1}(a)$ is stochastically equicontinuous, then since $\sqrt{nh} |R_{n,1}(a)|=o_p(1)$ for any $a \in \mathcal{A}$ we have 
$$ \sup_{a \in \mathcal{A}} | R_{n,1}(a)| = o_p\left( \sqrt{\log n / nh} \right) , $$
and so is asymptotically negligible. Finally the same logic as in Section 6.3 yields 
$$ \sup_{a \in \mathcal{A}} | R_{n,2}(a)| = O_p\left( \sup_{a \in \mathcal{A}} ||\hat\pi(a \mid \bL) - \pi(a \mid \bL) || \cdot ||\hat\mu(\bL,a) - \mu(\bL,a) || \right) , $$
so that for $\sup_{a \in \mathcal{A}} ||\hat\pi(a \mid \bL) - \pi(a \mid \bL) ||=O_p(r_n^*)$ and similarly for $\hat\mu$ and $s_n^*$ we have
$$ \sup_{a \in \mathcal{A}} | \hat\theta_h(a) - \theta(a) | = O_p\left( \sqrt{\frac{\log n}{nh}} + h^2 + r_n^* s_n^* \right) . $$

\pagebreak

\section{Sample R code}

\footnotesize
\begin{verbatim}
### INPUT: l is an n*p matrix, a and y are vectors of length n     
###  l = matrix of covariates
###  a = vector of treatment values
###  y = vector of observed outcomes

# set up evaluation points & matrices for predictions
a.min <- min(a); a.max <- max(a)
a.vals <- seq(a.min,a.max,length.out=100)
la.new <- rbind(cbind(l,a), cbind( l[rep(1:n,length(a.vals)),],
  a=rep(a.vals,rep(n,length(a.vals))) ))
l.new <- la.new[,-dim(la.new)[2]]

# fit super learner (other methods could be used here instead)
sl.lib <- c("SL.earth","SL.gam","SL.gbm","SL.glm","SL.glmnet")
pimod <- SuperLearner(Y=a, X=l, SL.library=sl.lib, newX=l.new)
pimod.vals <- pimod$SL.predict; sq.res <- (a-pimod.vals)^2
pi2mod <- SuperLearner(Y=sq.res,X=l, SL.library=sl.lib, newX=l.new)
pi2mod.vals <- pi2mod$SL.predict
mumod <- SuperLearner(Y=y, X=cbind(l,a), SL.library=sl.lib,
  newX=la.new,family=binomial); muhat.vals <- mumod$SL.predict

# construct estimated pi/varpi and mu/m values
approx.fn <- function(x,y,z){ predict(smooth.spline(x,y),x=x2)$y }
a.std <- (la.new$a-pimod.vals)/sqrt(pi2mod.vals)
pihat.vals <- approx.fn(density(a.std[1:n])$x, density(a.std[1:n])$y,
  a.std); pihat <- pihat.vals[1:n]
pihat.mat <- matrix(pihat.vals[-(1:n)], nrow=n,ncol=length(a.vals))
varpihat <- approx.fn(a.vals, apply(pihat.mat,2,mean), a)
varpihat.mat <- matrix(rep(apply(pihat.mat,2,mean),n), byrow=T,nrow=n)
muhat <- muhat.vals[1:n]
muhat.mat <- matrix(muhat.vals[-(1:n)], nrow=n,ncol=length(a.vals))
mhat <- approx.fn(a.vals, apply(muhat.mat,2,mean), a)
mhat.mat <- matrix( rep(apply(muhat.mat,2,mean),n), byrow=T,nrow=n)

# form adjusted/pseudo outcome xi
pseudo.out <- (y-muhat)/(pihat/varpihat) + mhat

# leave-one-out cross-validation to select bandwidth
library(KernSmooth); kern <- function(x){ dnorm(x) }
w.fn <- function(bw){ w.avals <- NULL; for (a.val in a.vals){
  a.std <- (a-a.val)/bw; kern.std <- kern(a.std)/bw
  w.avals <- c(w.avals, mean(a.std^2*kern.std)*(kern(0)/bw) / 
    (mean(kern.std)*mean(a.std^2*kern.std)-mean(a.std*kern.std)^2)) 
}; return(w.avals/n) }
hatvals <- function(bw){ approx(a.vals,w.fn(bw),xout=a)$y }
cts.eff <- function(out,bw){ approx(locpoly(a,out,bw),xout=a)$y }
# note: choice of bandwidth range depends on specific problem
h.opt <- optimize( function(h){ hats <- hatvals(h); 
    mean( ((pseudo.out - cts.eff(pseudo.out,bw=h))/(1-hats))^2) },
  c(0.01,50), tol=0.01)$minimum

# estimate effect curve with optimal bandwidth
est <- approx(locpoly(a,pseudo.out,bandwidth=h.opt),xout=a.vals)$y

# estimate sandwich-style pointwise confidence band
se <- NULL; for (a.val in a.vals){ 
a.std <- (a-a.val)/h.opt; kern.std <- (kern(a.std)/h.opt)/h.opt
beta <- coef(lm(pseudo.out ~ a.std, weights=kern.std))
Dh <- matrix( c(mean(kern.std), mean(kern.std*a.std),
  mean(kern.std*a.std), mean(kern.std*a.std^2)), nrow=2)
kern.mat <- matrix(rep(kern((a.vals-a.val)/h)/h,n), byrow=T,nrow=n)
g2 <- matrix( rep((a.vals-a.val)/h, n), byrow=T, nrow=n)
intfn1.mat <-  kern.mat*(muhat.mat - mhat.mat)*varpihat.mat
intfn2.mat <-  g2*kern.mat*(muhat.mat - mhat.mat)*varpihat.mat
int1 <- apply(matrix(rep((a.vals[-1]-a.vals[-length(a.vals)])/2,n), 
  byrow=T,nrow=n)*intfn1.mat[,-1]+intfn1.mat[,-length(a.vals)],1,sum)
int2 <- apply(matrix(rep((a.vals[-1]-a.vals[-length(a.vals)])/2,n), 
  byrow=T,nrow=n)*intfn2.mat[,-1]+intfn2.mat[,-length(a.vals)],1,sum)
sigma <- cov(t(solve(Dh) %*% 
  rbind( wt*(out-beta[1]-beta[2]*a.std) + int1, 
  a.std*wt*(out-beta[1]-beta[2]*a.std) + int2 )))
se <- c(se, sqrt(sigma[1,1])) }
ci.ll <- est-1.96*se/sqrt(n); ci.ul <- est+1.96*se/sqrt(n)
\end{verbatim}
\normalsize

\end{document}